\documentclass[preprint,12pt,authoryear]{elsarticle}

\usepackage[utf8]{inputenc}
\usepackage[T1]{fontenc}
\usepackage{lmodern}
\usepackage[english]{babel}
\usepackage{csquotes}
\usepackage{microtype}
\usepackage{enumitem}

\usepackage[margin=2cm]{geometry}
\usepackage{setspace}
\usepackage{amsmath, amssymb, amsthm}

\usepackage{booktabs}
\usepackage{longtable}
\usepackage{graphicx}
\usepackage{float}
\usepackage{multirow}
\usepackage{tabularx}
\newcolumntype{L}{>{\hsize=.82\hsize\linewidth=\hsize\raggedright\arraybackslash}X}
\newcolumntype{M}{>{\hsize=1.06\hsize\linewidth=\hsize\raggedright\arraybackslash}X}
\newcolumntype{R}{>{\hsize=1.12\hsize\linewidth=\hsize\raggedright\arraybackslash}X}

\theoremstyle{definition}
\newtheorem{thm}{Theorem}
\newtheorem{lem}{Lemma}

\newtheorem{prop}{Proposition}
\newtheorem{asm}{Assumption}
\newtheorem{rem}{Remark}
\newtheorem{ex}{Example}

\newcommand{\E}{\mathbb{E}}
\newcommand{\Prob}{\mathbb{P}}
\newcommand{\one}{\mathbf{1}}

\usepackage[dvipsnames]{xcolor}
\usepackage[colorlinks=true,
linkcolor=MidnightBlue,
citecolor=MidnightBlue,
urlcolor=MidnightBlue]{hyperref}
\hypersetup{pdfauthor={Junjie Li; Yukitoshi Matsushita}}

\journal{}
\makeatletter
\def\ps@pprintTitle{}
\makeatother

\begin{document}
\begin{frontmatter}
\title{Efficient difference-in-differences estimation under partial interference with incremental propensity score policies}

\author[hit]{Junjie Li}
\ead{junjieli.eco@gmail.com}

\author[hit]{Yukitoshi Matsushita\fnref{fund}}
\ead{matsushita.y@r.hit-u.ac.jp}

\address[hit]{Graduate School of Economics, Hitotsubashi University,
2-1 Naka, Kunitachi, Tokyo 186-8601, Japan.}


\begin{abstract}
This paper develops efficient difference-in-differences (DID) estimation under partial interference with a cluster incremental propensity score (CIPS) policy. We define direct and spillover average treatment effects on the treated, establish their identification, and derive their efficient influence functions, from which we construct a cross-fitted estimator. Simulations evaluate its finite-sample performance. An application to China's New Rural Pension Scheme recovers the reduction in farmwork among pension recipients reported by the original county-level analysis, separately estimates a within-household spillover alongside the direct effect, and traces both effects across the policy parameter.
\end{abstract}


\end{frontmatter}
\section{Introduction}
\label{sec:intro}
Difference-in-differences (DID) is one of the most widely used research designs for evaluating policy effects with panel or repeated cross-sectional data. Its standard identification argument relies on a no-interference (SUTVA) assumption, under which each unit's potential outcomes depend only on its own treatment status. In many empirical settings, however, this assumption is unlikely to hold. For example, when a place-based subsidy is granted to some firms, neighboring firms may benefit through agglomeration spillovers; when some children in a household are encouraged to attend school, their siblings may be affected through shared resources; when some villages in a county enter a special economic zone, surrounding villages may gain or lose through factor reallocation and local-market linkages; and when some members of a rural household begin to receive a social pension, their co-residents may adjust their labor supply through pooled household income, which is the setting of our empirical application. In these settings, a unit's outcome responds not only to its own treatment (a direct effect) but also to the treatment received by its peers (a spillover effect). Ignoring interference can therefore lead to a mischaracterization of the causal estimand and overlook spillover effects that are often of independent policy interest.

This paper studies efficient DID estimation under partial (clustered) interference: units are grouped into clusters, interference is unrestricted within a cluster but absent across clusters, and the number of clusters grows while cluster sizes stay bounded. We consider the data structure following \citet{lee2025efficient}. Clusters are assumed to be independent and identically distributed draws from the population, and the cluster size $N_i$ is a bounded random variable that is observed and incorporated throughout the analysis. 
Heterogeneity across clusters of different sizes is captured by conditioning the nuisance functions on the observed cluster size $N_i$, while allowing for arbitrary within-cluster dependence. This framework is more closely aligned with empirical applications discussed above.
In addition, we replace the type-B policy, under which each peer is independently assigned to treatment with the same probability $\alpha$, with a cluster incremental propensity score (CIPS) policy.
As \citet{lee2025efficient} emphasize, the type-B policy may lack practical
relevance because realistic counterfactual scenarios allow treatment
probabilities to vary across units according to their covariates.  Moreover,
when $\alpha$ differs substantially from the observed treatment rate, the
type-B estimand averages over peer treatment configurations that are rarely or
never observed.  A further difficulty is specific to the
difference-in-differences setting.  In randomized experiments the assignment
probability is fixed by design, so a counterfactual that sets it to a common
value $\alpha$ can be compared with the design probability.  In observational
panel data no such externally imposed probability exists, and it is unclear
which value of $\alpha$ should be regarded as close to the factual assignment
probabilities. To address these limitations, we adopt the cluster incremental propensity score (CIPS) policy of \citet{lee2025efficient}, which extends the incremental propensity score intervention of \citet{kennedy2019nonparametric} to clustered interference.

We make three contributions. First, we define policy-indexed direct and
spillover average treatment effects on the treated under partial interference
in a DID design. The CIPS policy traces these effects as the odds of
the observed, covariate-dependent treatment probabilities are shifted, while a
fixed type-B policy provides a useful benchmark. Under no anticipation and
configuration-specific conditional parallel trends, we establish identification
of both effects without imposing a parametric model on the treatment assignment
or outcome evolution. 
Second, we derive the efficient influence functions and semiparametric
efficiency bounds for the direct and spillover effects under both policies. The
ATT-type conditioning is consequential: its ratio normalization weights the
outcome-regression and policy-derivative components differently, so the
resulting influence functions cannot be obtained by simply replacing outcomes
with first differences in the ATE-type influence functions of
\citet{lee2025efficient}.  Under the CIPS policy they contain an
additional correction because the CIPS policy distribution is itself a
functional of the observed propensity scores; this correction vanishes under
the fixed type-B policy.
Third, we use these influence functions to construct cluster-level cross-fitted
estimators and establish asymptotic linearity, normality, consistent variance
estimation, and attainment of the semiparametric efficiency bounds. The fixed type-B estimator is doubly robust for consistency, whereas the CIPS estimator does not possess double robustness because misspecification of the propensity score changes the target policy itself, and consistency cannot be restored even with a correctly specified outcome model.  This result reproduces, for ATT, the one-sided robustness documented by \citet{lee2025efficient} for their ATE.

This work connects two strands of literature.  The first concerns causal
inference under interference and stochastic interventions.  Under partial
interference, \citet{hudgens2008toward} formalize direct, indirect, total, and
overall effects of treatment-allocation strategies, and
\citet{tchetgen2012causal} further develop causal estimands and
inverse-probability methods for settings with interference.  For observational
studies, \citet{liu2019doubly} construct estimators that remain consistent when
either the propensity-score or outcome-regression model is correctly
specified, while \citet{park2022efficient} characterize efficient influence
functions and efficiency bounds for a broad class of network treatment effects
under partial interference.  More recently, \citet{lee2025efficient} define a general class of
clustered stochastic policies, including policies that shift the observed
propensity-score distribution, and develop efficient nonparametric estimators
for their effects.
Our CIPS construction also draws on the broader stochastic-intervention
literature: \citet{munoz2012population} introduce population effects of
stochastically assigned exposures, \citet{kennedy2019nonparametric} develop
incremental interventions that shift the odds of the observed propensity
score, and \citet{wen2023intervention} study robustness for intervention
distributions that depend on the observed treatment process.

The second strand concerns DID with covariates and interference.
\citet{sant2020doubly} establish doubly robust and semiparametrically efficient
estimation of the ATT under conditional parallel trends without interference.
Allowing for interference, \citet{butts2023difference} studies treatment
assigned along geographic boundaries and traces part of the resulting bias to a
contaminated comparison group, whose outcomes are themselves affected by nearby
treatment.  The conditional parallel trends conditions (CPT-S1) and (CPT-S2) below formalize this distinction. 
\citet{xu2023difference} shows why canonical DID
estimators generally fail to identify direct and spillover effects and develops
doubly robust alternatives under modified parallel-trends conditions, whereas
\citet{sun2025difference} studies doubly robust estimation of direct and
spillover ATT parameters in networks after adjustment for high-dimensional
network confounders.  The present paper differs from this work in two respects.  First, these
papers target contrasts under fixed treatment allocations, whereas our estimands
are indexed by the CIPS parameter $\delta$ and therefore trace an entire policy
curve rather than a single number.  Second, they establish double robustness,
whereas we characterize the efficient influence function and the semiparametric
efficiency bound; to our knowledge, efficiency theory for difference-in-differences
under partial interference has not been available.  


The remainder of the paper is organized as follows.
Section~\ref{sec:setup} sets up the model, the CIPS policy, and the estimands.
Section~\ref{sec:identification} states the identifying assumptions and the
identification result, and illustrates them in a linear model with peer
effects.  Section~\ref{sec:efficiency} derives the efficient influence
functions and the semiparametric efficiency bounds, and develops the
cross-fitted estimator and its asymptotic theory.
Section~\ref{sec:simulation} reports a simulation study and
Section~\ref{sec:application} presents the empirical application;
Section~\ref{sec:discussion} concludes.  \ref{app:notation} collects the
notation and \ref{app:proofs} the proofs.

\section{Set-up}
\label{sec:setup}

\subsection{Data structure}

We first introduce the cluster-sampling framework and notation that will be used to define the causal estimands and derive their identification results. Suppose that $M$ observable clusters are independently sampled from a super-population $\Prob$. Cluster $i$ has a random size $N_i\in\mathbb{N}$ and contains units $j=1,\dots,N_i$. We focus on two time periods $t=1,2$, with treatment realized only in period $2$. Following the DID convention, we write the period-2 treatment as $A_{ij}\equiv A_{ij2}\in\{0,1\}$. For unit $j$ in cluster $i$, let $Y_{ijt}\in\mathbb{R}$ be the outcome at time $t$ and $X_{ij}\in\mathbb{R}^p$ a vector of pre-treatment covariates. Stacking units within a cluster,
\[
Y_{it}=(Y_{i1t},\dots,Y_{iN_it})^\top,\quad
A_i=(A_{i1},\dots,A_{iN_i})^\top\in\mathcal{A}(N_i),\quad
X_i=(X_{i1}^\top,\dots,X_{iN_i}^\top)^\top,
\]
where $\mathcal{A}(s)$ denotes the set of all length-$s$ binary vectors. For unit $j$ in cluster $i$, let
$A_{i(-j)}=(A_{ik})_{k\ne j,\,1\le k\le N_i}$ denote the treatment vector of $j$'s peers, so that 
$A_i=(A_{ij},A_{i(-j)})$ after a relabeling of coordinates.  Write
$\mathbf 0_{-j}\in\mathcal A(N_i-1)$ for the all-zero peer-treatment vector.
The observed data for cluster $i$ are
\[
O_i=(Y_{i1},Y_{i2},A_i,X_i,N_i),
\]
and we assume $(O_1,\dots,O_M)$ are independent and identically distributed draws from $\Prob$, and that there exists a fixed $n_{\max}\in\mathbb{N}$ with $\Prob(N_i\le n_{\max})=1$. Asymptotics are driven by $M\to\infty$ with the cluster size bounded, following \citet{liu2019doubly,park2022efficient,lee2025efficient}. For ease of notation we write
\[
\nu_n=\Prob(N_i=n),\qquad \E_n[\cdot]=\E[\cdot\mid N_i=n],\qquad n\le n_{\max}.
\]
For unit $j$, define the observed first difference $\Delta Y_{ij}=Y_{ij2}-Y_{ij1}$ and the reduced data
\[
W_{ij}=(\Delta Y_{ij},A_{ij},A_{i(-j)},X_i,N_i);
\]
the analysis depends on the outcomes only through $\Delta Y_{ij}$.
This representation isolates the cluster-level observed
data on which the analysis is based.  We next embed these variables in a
potential-outcome model and specify the counterfactual policy governing peer
treatments.

\subsection{Potential outcomes and the CIPS policy}
\label{sec:cips}

Building on the observed-data structure above, we now
describe how outcomes may depend on the entire within-cluster treatment
vector and how that vector is shifted by the policy.
We use the potential-outcome framework for partial interference \citep{hudgens2008toward,tchetgen2012causal,liu2019doubly,park2022efficient,lee2025efficient}. Let $Y_{ijt}(a_i)=Y_{ijt}(a_{ij},a_{i(-j)})$ be the potential outcome of unit $j$ at time $t$ when cluster $i$ receives $a_i=(a_{ij},a_{i(-j)})$; a unit's potential outcome may depend on the treatments of others in the same cluster but not on those in other clusters.
Let
$q_{j}(x,n)=\Prob(A_{ij}=1\mid X_i=x,N_i=n)$
denote the individual propensity score of unit $j$. The incremental propensity score shifts the odds of each unit's treatment by a user-chosen factor $\delta\in(0,\infty)$, producing a supported family of counterfactual peer allocations indexed by $\delta$:  
\begin{equation}\label{eq:qdelta}
q_{j,\delta}(x,n)=\frac{\delta\,q_{j}(x,n)}{\delta\,q_{j}(x,n)+1-q_{j}(x,n)},
\qquad\text{equivalently}\qquad
\frac{q_{j,\delta}(x,n)}{1-q_{j,\delta}(x,n)}=\delta\,\frac{q_{j}(x,n)}{1-q_{j}(x,n)}.
\end{equation}
The choice $\delta=1$ recovers the factual propensity score, while $\delta\downarrow0$ shifts all units toward control and $\delta\uparrow\infty$ shifts them toward treatment. Moreover, the ranking of units according to their treatment probabilities is preserved. The incremental propensity score therefore transforms the factual propensity scores into a supported family of counterfactual peer treatment allocations indexed by $\delta$, following the framework of \citet{kennedy2019nonparametric}. Following \citet{lee2025efficient}, we assume that treatment assignments of distinct units within the same cluster are conditionally independent given $(X_i,N_i)$, which leads to the CIPS policy:
\begin{equation}\label{eq:assignCI}
\Prob(A_i=a_i\mid X_i=x,N_i=n)
=\prod_{k=1}^{n}q_k(x,n)^{a_k}\{1-q_k(x,n)\}^{1-a_k}.
\end{equation}
By shifting each individual propensity score in this product, we obtain the conditional CIPS policy distribution:
\begin{equation}\label{eq:pidelta}
\pi_\delta(a_{-j}\mid x,n)
=\prod_{j'\neq j}q_{j',\delta}(x,n)^{a_{j'}}\bigl(1-q_{j',\delta}(x,n)\bigr)^{1-a_{j'}},
\qquad a_{-j}\in\mathcal{A}(n-1).
\end{equation}
The CIPS policy
is constructed as the product of the individual incremental propensity scores, so it does not rely on any parametric models. Specifically, it relies on the assumption that treatment assignments between units within the same cluster are conditionally independent given covariates and cluster size, while allowing for marginal dependence induced by shared covariates. This conditional independence assumption is weaker than the type-B policy assumption (as discussed in Remark~\ref{rem:typeB}), which requires all units to share the same treatment probability $\alpha$ and to be marginally independent.

\subsection{The causal estimands}
\label{sec:estimands}

With the potential outcomes and the CIPS policy distribution defined above, we now introduce the causal estimands. Fix a focal unit $j$ in a cluster of size $n$. For a peer configuration $a_{-j}\in\mathcal{A}(n-1)$ and a covariate value $x$ in the support of $X_i\mid N_i=n$, we define the conditional direct and spillover average treatment effects on the treated as follows:
\begin{align}
\theta_{j}^{DATT}(a_{-j},x,n)
&=\E\!\left[Y_{ij2}(1,a_{-j})-Y_{ij2}(0,a_{-j})\mid A_{ij}=1,A_{i(-j)}=a_{-j},X_i=x,N_i=n\right],
\label{eq:thetaDATTx}\\
\theta_{j}^{SATT}(a_{-j},x,n)
&=\E\!\left[Y_{ij2}(0,a_{-j})-Y_{ij2}(0,\mathbf 0_{-j})\mid A_{ij}=1,A_{i(-j)}=a_{-j},X_i=x,N_i=n\right].
\label{eq:thetaSATTx}
\end{align}
Both estimands condition not only on the covariates $X_i=x$ but also on the realized peer configuration $A_{i(-j)}=a_{-j}$. Therefore, the subpopulation over which $\theta_{j}^{\bullet}(a_{-j},x,n)$ is defined varies with $a_{-j}$. The aggregation in \eqref{eq:tauj} consequently combines configuration-specific treated subpopulations, which is why the parallel-trends restrictions in Assumption~\ref{asm:1}(iii) are imposed separately for each peer configuration. The parameters of interest are the direct and spillover effects under the policy:
\begin{equation}\label{eq:tau}
\tau^{\bullet}(\delta)
=\E\!\left[N_i^{-1}\sum_{j=1}^{N_i}\tau_{j}^{\bullet}(\delta,N_i)\right]
=\sum_{n\le n_{\max}}\nu_n\,\frac1n\sum_{j=1}^{n}\tau_{j}^{\bullet}(\delta,n),
\  \bullet\in\{DATT,SATT\},
\end{equation}
where 
\begin{equation}\label{eq:tauj}
\tau_{j}^{\bullet}(\delta,n)
=\sum_{a_{-j}\in\mathcal{A}(n-1)}
\E\!\left[\pi_\delta(a_{-j}\mid X_i,n)\,\theta_{j}^{\bullet}(a_{-j},X_i,n)\;\middle|\;A_{ij}=1,N_i=n\right],
\  \bullet\in\{DATT,SATT\},
\end{equation}
denotes the unit-level effect under the policy.
In \eqref{eq:tauj}, the covariate-specific weight
$\pi_\delta(a_{-j}\mid X_i,n)$ and the conditional treatment effect
$\theta_{j}^{\bullet}(a_{-j},X_i,n)$ are multiplied pointwise in $X_i$, and
the product is averaged over the distribution of $X_i$ given
$A_{ij}=1,N_i=n$.  The estimand \eqref{eq:tau} therefore weights each
treated unit by its realized covariate profile, while the summation over $a_{-j}$ aggregates across configuration-specific treated subpopulations described above.

The conditional and unconditional estimands satisfy an exact decomposition of a total effect: writing
$\theta_{j}^{TATT}(a_{-j},x,n)=\theta_{j}^{DATT}(a_{-j},x,n)+\theta_{j}^{SATT}(a_{-j},x,n)$
and aggregating with the common weights in \eqref{eq:tau}--\eqref{eq:tauj} gives
$\tau^{TATT}(\delta)=\tau^{DATT}(\delta)+\tau^{SATT}(\delta)$ for every $\delta\in(0,\infty)$. In the no-interference special case $N_i\equiv1$, every cluster contains a single unit with no peers, $\mathcal{A}(0)=\{\emptyset\}$, $\pi_\delta(\emptyset\mid x,1)=1$, and $Y_{ij2}(a_{ij},a_{i(-j)})=Y_{ij2}(a_{ij})$; then $\tau^{DATT}(\delta)$ reduces, for every $\delta$, to the standard average treatment effect on the treated of \citet{hahn1998role, abadie2005semiparametric, sant2020doubly}, and $\tau^{SATT}(\delta)\equiv0$.
These policy-indexed estimands complete the set-up by
linking the configuration-specific causal contrasts to the observed
covariate distribution.

\subsection{Relation to the literature}
\label{sec:literature}

Much of the recent literature on DID under interference
uses exposure mappings to summarize how the treatments of others affect a focal unit and defines causal estimands accordingly. The estimands defined above depart from this
literature in that peer treatment assignments are not summarized through an
exposure mapping. Instead, the complete ordered peer vector $a_{-j}$ is retained in \eqref{eq:thetaDATTx}--\eqref{eq:thetaSATTx}, and the CIPS policy $\pi_\delta(a_{-j}\mid x,n)$ in \eqref{eq:pidelta} assigns a covariate-dependent distribution over peer configurations in \eqref{eq:tauj}, rather than fixing a single exposure level.
Table~\ref{tab:did-interference-comparison} compares the interference setting, exposure mapping, DATT definition, and parallel-trends assumption in the present paper with those in \citet{hettinger2025doubly}, \citet{lee2026policy}, and \citet{xu2023difference}.

\begin{table}[htbp]
\centering
\setstretch{1}
\caption{Comparison with difference-in-differences methods under interference}
\label{tab:did-interference-comparison}
\setlength{\abovecaptionskip}{0pt}
\setlength{\belowcaptionskip}{6pt}
\scriptsize
\setlength{\tabcolsep}{2pt}
\renewcommand{\arraystretch}{1.25}
\begin{tabularx}{\linewidth}{@{}>{\raggedright\arraybackslash}p{2.1cm}LMR@{}}
\toprule
\textbf{Study / interference} & \textbf{Exposure mapping} &
\textbf{DATT definition} & \textbf{Parallel trends} \\
\midrule

\citet{hettinger2025doubly}\newline
\textit{Geographically separable}
& $G_j=1$ if either (i) $a_j=0$ and unit $j$ is adjacent to a treated region,
  or (ii) $a_j=1$ and unit $j$ is not adjacent to an untreated region;
  otherwise $G_j=0$.
& {\fontsize{6.5}{7.5}\selectfont \(\begin{aligned}
  &\E[Y_{j2}(1,0)-Y_{j2}(0,0)\mid A_j=1,G_j=0]\\[1pt]
  ={}&\E[Y_{j2}(1,0)-Y_{j2}(0,0)\mid A_j=1]
  \end{aligned}\)}
& {\fontsize{6}{7}\selectfont \(\begin{aligned}
  &\E[Y_{j2}(0,0)\!-\!Y_{j1}(0,0)\!\mid\! A_j=1,G_j=0,X_j]\\[-1pt]
  ={}&\E[Y_{j2}(0,0)\!-\!Y_{j1}(0,0)\!\mid\! A_j=0,G_j=0,X_j]
  \end{aligned}\)}
\\
\addlinespace[5pt]

\citet{lee2026policy}\newline
\textit{Partial interference}
& $G_{ij}=1$ if either (i) $a_{ij}=0$ and unit $j$ has at least one treated
  neighbour in cluster $i$, or (ii) $a_{ij}=1$ and all of its neighbours in
  cluster $i$ are treated; otherwise $G_{ij}=0$.
& {\fontsize{6.5}{7.5}\selectfont \(\begin{aligned}
  &\E[Y_{ij2}(1,0)-Y_{ij2}(0,0)\mid A_{ij}=1,G_{ij}=0]\\[1pt]
  ={}&\E[Y_{ij2}(1,0)-Y_{ij2}(0,0)\mid A_{ij}=1]
  \end{aligned}\)}
& {\fontsize{6}{7}\selectfont \(\begin{aligned}
  &\E[Y_{ij2}(0,0)\!-\!Y_{ij1}(0,0)\!\mid\! A_{ij}=1,G_{ij}=0,X_{ij}]\\[-1pt]
  ={}&\E[Y_{ij2}(0,0)\!-\!Y_{ij1}(0,0)\!\mid\! A_{ij}=0,G_{ij}=0,X_{ij}]
  \end{aligned}\)}
\\
\addlinespace[5pt]

\citet{xu2023difference}\newline
\textit{General interference}
& $G_j=g_j(A_{-j})\in\mathcal G$.
& {\fontsize{6.5}{7.5}\selectfont \(\begin{aligned}
  &|J|^{-1}\!\sum_{j\in J}\sum_{g\in\mathcal G} \Pr(G_j=g\!\mid\! A_j=1,x) \times \\
  &\E[Y_{j2}(1,g)-Y_{j2}(0,g) \mid A_j=1,G_j=g,x]
  \end{aligned}\)}
& {\fontsize{6}{7}\selectfont \(\begin{aligned}
  &\E[Y_{j2}(0,g)\!-\!Y_{j1}(0,0)\!\mid\! A_j=1,G_j=g,x]\\[-1pt]
  ={}&\E[Y_{j2}(0,g)\!-\!Y_{j1}(0,0)\!\mid\! A_j=0,G_j=g,x]
  \end{aligned}\)}
\\
\addlinespace[5pt]

\textbf{This paper}\newline
\textit{Partial interference}
& \textbf{None.}  The complete ordered peer vector $a_{-j}\in\mathcal A(n-1)$
  is retained; the CIPS policy $\pi_\delta(a_{-j}\mid x,n)$ of
  \eqref{eq:pidelta} places a covariate-dependent distribution over
  configurations.
& The conditional direct effect $\theta_j^{DATT}(a_{-j},x,n)$ of
  \eqref{eq:thetaDATTx}, aggregated over peer configurations by
  \eqref{eq:tauj} and over cluster sizes by \eqref{eq:tau} to give
  $\tau^{DATT}(\delta)$.
& The conditional parallel-trends condition for the DATT in
  Assumption~\ref{asm:1}(iii) below, imposed separately for each peer
  configuration $a_{-j}$.
\\

\bottomrule
\end{tabularx}
\vspace{0.5em}

\begin{minipage}{\linewidth}
\fontsize{7}{8}\selectfont
\setlength{\parindent}{0pt}
\textit{Note:}
Notation is rewritten in the notation of this paper.
In the first and third rows $j$ indexes a unit in the whole population and no
cluster index is used; in the second and fourth rows $ij$ indexes unit $j$ of
cluster $i$.
$\mathcal G$ is a discrete set of possible exposure values and $|J|$ is the
finite population size.
The second equality in the first two rows is specific to the data used in
those studies: no unit is observed that is itself treated and has only treated
neighbours, so that $A_j=1$ implies $G_j=0$ and $A_{ij}=1$ implies
$G_{ij}=0$.
\end{minipage}
\end{table}

Section~\ref{sec:identification} next states the
conditions under which each component, and hence the aggregated effects, is
identified from the observed data.

\section{Identification}
\label{sec:identification}

To identify the DATT and SATT, we introduce the following outcome regression and treatment assignment nuisance functions. For $a\in\{0,1\}$,
\begin{align*}
m_{a,j}(a_{-j},x,n)
&=\E[\Delta Y_{ij}\mid A_{ij}=a,A_{i(-j)}=a_{-j},X_i=x,N_i=n],\\
\Delta m_{j}(a_{-j},x,n)
&=m_{1,j}(a_{-j},x,n)-m_{0,j}(a_{-j},x,n),
\end{align*}
\vspace{-2cm}
\begin{align*}
q_{j}(x,n)
&=\Prob(A_{ij}=1\mid X_i=x,N_i=n),\\
e_{a,j}(a_{-j}\mid x,n)
&=\Prob(A_{i(-j)}=a_{-j}\mid A_{ij}=a,X_i=x,N_i=n),
\end{align*}
where the conditional independence restriction in \eqref{eq:assignCI} implies that
$e_{0,j}(a_{-j}\mid x,n)=e_{1,j}(a_{-j}\mid x,n)
=\prod_{j'\ne j}\left[
q_{j'}(x,n)^{a_{j'}}\{1-q_{j'}(x,n)\}^{1-a_{j'}}
\right]$.
The marginal probability of unit $j$ is treated within clusters of size-$n$ is
$p_{j}(n)=\Prob(A_{ij}=1\mid N_i=n)=\E_n[q_{j}(X_i,n)]$.
Note that the covariate distribution among treated units can be recovered by Bayes' rule: for any integrable $f$,
\begin{equation}\label{eq:bayes}
\E[f(X_i)\mid A_{ij}=1,N_i=n]=\frac{\E_n[q_{j}(X_i,n)\,f(X_i)]}{p_{j}(n)}.
\end{equation}
 We next state the assumptions under which the conditional and unconditional DATT and SATT are identified from the observed data.
\begin{asm}[Identification assumptions]\label{asm:1}
For all $n\le n_{\max}$, $j\le n$, $a_{-j}\in\mathcal{A}(n-1)$, and $x$ in the support of $X_i\mid N_i=n$, the following hold.
\begin{enumerate}[label=(\roman*)]
\item (Consistency)
$Y_{ijt}=\sum_{a_i\in\mathcal{A}(N_i)}\one\{A_i=a_i\}\,Y_{ijt}(a_i)$ for all $i,j,t$.
\item (No anticipation) Treatment occurs only in period $2$, and
$Y_{ij1}(a_{ij},a_{i(-j)})=Y_{ij1}(0,\mathbf 0_{-j})$ for all $i,j$.
\item (Conditional parallel trends) For the identification of $\tau^{DATT}(\delta)$,
\begin{align*}
&\E[Y_{ij2}(0,a_{-j})-Y_{ij1}(0,\mathbf 0_{-j})\mid A_{ij}=1,A_{i(-j)}=a_{-j},X_i=x,N_i=n]\\
&\quad=\E[Y_{ij2}(0,a_{-j})-Y_{ij1}(0,\mathbf 0_{-j})\mid A_{ij}=0,A_{i(-j)}=a_{-j},X_i=x,N_i=n].
\end{align*}
For the identification of $\tau^{SATT}(\delta)$, the following two conditions hold:

\noindent (CPT-S1, standard parallel trends for $Y(0,a_{-j})$)
\begin{align*}
&\E[Y_{ij2}(0,a_{-j})-Y_{ij1}(0,a_{-j})\mid A_{ij}=1,A_{i(-j)}=a_{-j},X_i=x,N_i=n]\\
&\quad=\E[Y_{ij2}(0,a_{-j})-Y_{ij1}(0,a_{-j})\mid A_{ij}=0,A_{i(-j)}=a_{-j},X_i=x,N_i=n].
\end{align*}

\noindent (CPT-S2, cross-configuration trend for $Y(0,\mathbf 0_{-j})$)
\begin{align*}
&\E[Y_{ij2}(0,\mathbf 0_{-j})-Y_{ij1}(0,\mathbf 0_{-j})\mid A_{ij}=1,A_{i(-j)}=a_{-j},X_i=x,N_i=n]\\
&\quad=\E[Y_{ij2}(0,\mathbf 0_{-j})-Y_{ij1}(0,\mathbf 0_{-j})\mid A_{ij}=0,A_{i(-j)}=\mathbf 0_{-j},X_i=x,N_i=n].
\end{align*}
CPT-S1 is the usual conditional parallel-trends assumption within peer configuration $a_{-j}$; CPT-S2 is a cross-configuration restriction stating that treated units with peer configuration $a_{-j}$ share the same $Y(0,\mathbf 0_{-j})$ trend of control units that have no treated peers.
\item (Overlap) There exists a constant $c>0$, not depending on
$(j,n,x)$, such that
\[
c\le q_{j}(x,n)\le 1-c.
\]
Under \eqref{eq:assignCI} this single restriction implies
\[
e_{a,j}(a_{-j}\mid x,n)\ge c^{\,n_{\max}-1}>0
\quad(a\in\{0,1\}),\qquad
p_{j}(n)=\E_n[q_{j}(X_i,n)]\ge c>0,
\]
for every $a_{-j}\in\mathcal{A}(n-1)$, in particular all $2^{\,n-1}$
peer configurations have positive conditional probability, so no separate
support qualification is needed.
\item (Moments)
\[
\E_n[\Delta Y_{ij}^{2}]<\infty,\qquad
\E_n\!\left[\max_{a\in\{0,1\},\,a_{-j}\in\mathcal{A}(n-1)}|m_{a,j}(a_{-j},X_i,n)|^{2}\right]<\infty,
\]
and there is a constant $C_Y<\infty$ with
$\E[\Delta Y_{ij}^{2}\mid A_i,X_i,N_i]\le C_Y$ almost surely.
\end{enumerate}
\end{asm}

\begin{prop}[Identification]\label{prop:id}
Under Assumption~\ref{asm:1}, for every $\delta\in(0,\infty)$ the covariate-conditional direct and spillover average treatment effects on the treated are identified as
\begin{align}
\theta_{j}^{DATT}(a_{-j},x,n)
&=m_{1,j}(a_{-j},x,n)-m_{0,j}(a_{-j},x,n)=\Delta m_{j}(a_{-j},x,n),
\label{eq:idDATT}\\
\theta_{j}^{SATT}(a_{-j},x,n)
&=m_{0,j}(a_{-j},x,n)-m_{0,j}(\mathbf 0_{-j},x,n).
\label{eq:idSATT}
\end{align}
Consequently, using \eqref{eq:bayes}, the unit-level effects \eqref{eq:tauj} are identified as
\begin{align}
\tau_{j}^{DATT}(\delta,n)
&=\frac{1}{p_{j}(n)}\sum_{a_{-j}\in\mathcal{A}(n-1)}
\E_n\!\left[q_{j}(X_i,n)\,\pi_\delta(a_{-j}\mid X_i,n)\,\Delta m_{j}(a_{-j},X_i,n)\right],
\label{eq:idtauDATT}\\
\tau_{j}^{SATT}(\delta,n)
&=\frac{1}{p_{j}(n)}\sum_{a_{-j}\in\mathcal{A}(n-1)}
\E_n\!\left[q_{j}(X_i,n)\,\pi_\delta(a_{-j}\mid X_i,n)\,\bigl\{m_{0,j}(a_{-j},X_i,n)-m_{0,j}(\mathbf 0_{-j},X_i,n)\bigr\}\right],
\label{eq:idtauSATT}
\end{align}
and hence $\tau^{DATT}(\delta)$ and $\tau^{SATT}(\delta)$ in \eqref{eq:tau} are identified. A proof is provided in~\ref{app:id}.
\end{prop}

\begin{ex}[A linear DID model with peer effects]\label{ex:linear}
It is instructive to read the estimands \eqref{eq:tau}--\eqref{eq:tauj} through
a linear DID specification.  Throughout this example we maintain
Assumption~\ref{asm:1} and impose the following linear model as an additional,
illustrative parametric restriction.  For simplicity we omit covariate terms
from the outcome equation and write, for $t\in\{1,2\}$ and
$\mathrm{Post}_t=\one\{t=2\}$,
\begin{equation}\label{eq:exlevels}
Y_{ijt}=\alpha_{ij}+\lambda_t
+\beta^{D}\bigl(A_{ij}\times\mathrm{Post}_t\bigr)
+\sum_{k\ne j}\beta^{S}_{jk}\bigl(A_{ik}\times\mathrm{Post}_t\bigr)+u_{ijt},
\end{equation}
where $\alpha_{ij}$ is a unit effect and $\lambda_t$ a common time effect.
Equation \eqref{eq:exlevels} is read as structural: it continues to hold with
$A_i$ replaced by any $a_i\in\mathcal A(N_i)$, and consistency
(Assumption~\ref{asm:1}(i)) returns the displayed form at the realized
assignment.  The specification excludes the treatment interactions
$A_{ij}\times A_{ik}$ and $A_{ik}\times A_{ik'}$, and takes
$(\beta^{D},\beta^{S}_{jk})$ to be homogeneous across clusters.  It also keeps
the additive two-way fixed effect structure, in contrast to the interactive fixed
effects $\alpha_{ij}\times\lambda_t$ allowed by
\citet{callaway2023treatment}.

Because the treatment terms are multiplied by $\mathrm{Post}_t$, the period-$1$
equation carries no treatment, and first differencing eliminates $\alpha_{ij}$:
$\Delta Y_{ij}=\Delta\lambda+\beta^{D}A_{ij}+\sum_{k\ne j}\beta^{S}_{jk}A_{ik}
+\Delta u_{ij}$ with $\Delta\lambda=\lambda_2-\lambda_1$.  Assumption~\ref{asm:1}(iii)
then translates the parallel-trends restrictions into restrictions on the
conditional mean of $\Delta u_{ij}$.  For the direct contrast, the DATT
parallel-trends condition equates the conditional means of $\Delta u_{ij}$ in
the cells $(A_{ij},A_{i(-j)})=(1,a_{-j})$ and $(0,a_{-j})$.  For the spillover
contrast, CPT-S1 gives the same equality between these two cells, while CPT-S2
equates the first cell to the no-treated-peer control cell
$(0,\mathbf 0_{-j})$.  Thus the relevant error trend is common across the
cells used by each contrast, and it cancels from the contrasts identified in
Proposition~\ref{prop:id}, giving
\begin{equation}\label{eq:extheta}
\theta_{j}^{DATT}(a_{-j},n)=\beta^{D},
\qquad
\theta_{j}^{SATT}(a_{-j},n)=\sum_{k\ne j}\beta^{S}_{jk}a_{k}.
\end{equation}
The direct effect is the own-treatment coefficient, and the spillover effect
adds up the focal-unit-specific peer coefficients at the realized
configuration.
The policy enters only at the aggregation step.  For every
$\delta\in(0,\infty)$, substituting \eqref{eq:extheta} into
\eqref{eq:tauj} yields
\begin{align}
\tau_{j}^{DATT}(\delta,n)
&=
\sum_{a_{-j}\in\mathcal A(n-1)}
\E\!\left[
\pi_\delta(a_{-j}\mid X_i,n)\,\beta^{D}
\;\middle|\; A_{ij}=1,N_i=n
\right]
=\beta^{D},
\label{eq:exunitdatt}\\
\tau_{j}^{SATT}(\delta,n)
&=
\sum_{a_{-j}\in\mathcal A(n-1)}
\E\!\left[
\pi_\delta(a_{-j}\mid X_i,n)
\sum_{k\ne j}\beta^{S}_{jk}a_k
\;\middle|\; A_{ij}=1,N_i=n
\right]\nonumber\\
&=
\sum_{k\ne j}\beta^{S}_{jk}
\E\!\left[
q_{k,\delta}(X_i,n)
\;\middle|\; A_{ij}=1,N_i=n
\right].
\label{eq:exunitsatt}
\end{align}
The last equality follows because the marginal probability that peer $k$ is
treated under $\pi_\delta(\cdot\mid X_i,n)$ is
$q_{k,\delta}(X_i,n)$.  At $\delta=1$,
$q_{k,1}(X_i,n)=q_k(X_i,n)$, so the policy reproduces the factual
conditional peer-treatment allocation.
If the propensity is homogeneous,
$q_j(x,n)\equiv q$ for every $(j,x,n)$, so that
$q_{k,\delta}\equiv q_\delta=\delta q/\{\delta q+1-q\}$, and the spillover
coefficients are homogeneous across all focal-peer pairs,
$\beta^{S}_{jk}\equiv\beta^{S}$ for every $j\ne k$, then
\eqref{eq:exunitdatt}--\eqref{eq:exunitsatt} become
\[
\tau_{j}^{DATT}(\delta,n)=\beta^{D},
\qquad
\tau_{j}^{SATT}(\delta,n)=(n-1)\beta^{S}q_\delta.
\]
Consequently, the global effects in \eqref{eq:tau} reduce to
\begin{equation}\label{eq:extau}
\tau^{DATT}(\delta)=\beta^{D},
\qquad
\tau^{SATT}(\delta)=\beta^{S}\,\E[N_i-1]q_\delta.
\end{equation}
Thus the direct effect is the own-treatment coefficient and is invariant to
the policy strength.  The spillover effect equals the peer coefficient times
the policy-induced expected number of treated peers: it converges to zero as
$\delta\downarrow0$, equals $\beta^{S}\E[N_i-1]q$ at $\delta=1$, and converges
to $\beta^{S}\E[N_i-1]$ as $\delta\uparrow\infty$.  When $N_i\equiv1$, the peer
block is empty, $\tau^{DATT}(\delta)=\beta^{D}$ is the ATT, and
$\tau^{SATT}(\delta)\equiv0$, as expected because no
peers are present. The two homogeneity restrictions used for \eqref{eq:extau} are
imposed only to obtain a closed form, and the simulation design of
Section~\ref{sec:simulation} dispenses with both: there the propensity $q_j$
depends on the unit label, the cluster size, and the entire cluster covariate
vector, and the peer coefficient $\beta^{S}_{jk}$ varies with the peer's
position $k$.
\end{ex}

\section{Semiparametric efficiency}
\label{sec:efficiency}

Let $\mathcal{M}_{NP}$ denote the nonparametric model for the i.i.d.\ cluster
data of Section~\ref{sec:setup}, subject to the conditional independence
restriction \eqref{eq:assignCI}. Throughout, starred quantities denote
evaluation at the truth. We first record the generic policy derivative that
carries the CIPS weight's dependence on the assignment mechanism,
\begin{align}\label{eq:gprime}
g_{j,a_{-j},j'}^{\prime}(x,n)
=\frac{\partial \pi_\delta(a_{-j}\mid x,n)}{\partial q_{j'}(x,n)}
=\frac{a_{j'}-q_{j',\delta}(x,n)}{q_{j',\delta}(x,n)\{1-q_{j',\delta}(x,n)\}}
\frac{\pi_\delta(a_{-j}\mid x,n) \, \delta}{\{\delta q_{j'}(x,n)+1-q_{j'}(x,n)\}^{2}},
\end{align}
the last factor being $\pi_\delta(a_{-j}\mid x,n) \, \partial q_{j',\delta}/\partial q_{j'}$ obtained from
\eqref{eq:qdelta}.  Write
$g_{j,a_{-j},j'}^{\prime*}$ for \eqref{eq:gprime} evaluated at $q^*$.
This derivative is the object that generates the influence function
$\phi_Q$ of the policy weight in \citet{lee2025efficient}, see Section~B.2 of
their supplementary material for further details.
We use the superscript $\bullet$ to denote either $\{DATT,SATT\}$ for notational convenience. Define the policy averages of the
regression contrasts within a size-$n$ cluster as
\begin{align*}
H_j^{DATT*}(x,n)
&=\sum_{a_{-j}\in\mathcal A(n-1)}
\pi_\delta^*(a_{-j}\mid x,n)\Delta m_j^*(a_{-j},x,n),\\
H_j^{SATT*}(x,n)
&=\sum_{a_{-j}\in\mathcal A(n-1)}
\pi_\delta^*(a_{-j}\mid x,n)
\{m_{0,j}^*(a_{-j},x,n)-m_{0,j}^*(\mathbf 0_{-j},x,n)\},
\end{align*}
together with the corresponding policy-estimation terms:
{\small
\begin{align}
G_{j}^{DATT*}(O_i;\delta,n)
&=q_{j}^*(X_i,n)\sum_{a_{-j}\in\mathcal{A}(n-1)}\Delta m_{j}^*(a_{-j},X_i,n)
\sum_{j'\neq j}g_{j,a_{-j},j'}^{\prime*}(X_i,n)\bigl(A_{ij'}-q_{j'}^*(X_i,n)\bigr),
\label{eq:Gdatt}\\
G_{j}^{SATT*}(O_i;\delta,n)
&=q_{j}^*(X_i,n)\sum_{a_{-j}\in\mathcal{A}(n-1)}\bigl\{m_{0,j}^*(a_{-j},X_i,n)-m_{0,j}^*(\mathbf 0_{-j},X_i,n)\bigr\} \sum_{j'\neq j}g_{j,a_{-j},j'}^{\prime*}(X_i,n)\bigl(A_{ij'}-q_{j'}^*(X_i,n)\bigr).
\label{eq:Gsatt}
\end{align}
}

Exchanging the sum over $a_{-j}$ with the expectation in
\eqref{eq:idtauDATT}--\eqref{eq:idtauSATT} shows that $H_j^{\bullet*}$ is the
inner policy average appearing in the identification result, so that the
unit-level estimand can be written compactly as
$\tau_j^{\bullet*}(\delta,n)
=\E_n\bigl[q_j^*(X_i,n)H_j^{\bullet*}(X_i,n)\bigr]/p_j^*(n)$.
That is, $\tau_j^{\bullet*}(\delta,n)$ is the $q_j^*$-weighted mean of
$H_j^{\bullet*}$ within size-$n$ clusters, normalized by the treated marginal
$p_j^*(n)$.

\begin{thm}[Global efficiency under the CIPS policy]\label{thm:eif}
Under Assumption~\ref{asm:1}, for $\bullet\in\{DATT,SATT\}$ the efficient influence function of $\tau^{\bullet}(\delta)$ under $\mathcal{M}_{NP}$ is
\begin{equation}\label{eq:eif}
\varphi^{\bullet*}(O_i;\delta)
=\bar\varphi^{\bullet*}(O_i;\delta)+\Bigl\{\bar\tau^{\bullet*}(\delta,N_i)-\tau^{\bullet*}(\delta)\Bigr\},
\quad
\bar\tau^{\bullet*}(\delta,n)=n^{-1}\sum_{j=1}^{n}\tau_{j}^{\bullet*}(\delta,n),
\end{equation}
where $\bar\varphi^{\bullet*}(O_i;\delta)=N_i^{-1}\sum_{j=1}^{N_i}\varphi_{j}^{\bullet*}(O_i;\delta,N_i)$ and the unit-level EIF is
\begin{align}
\varphi_{j}^{DATT*}(O_i;\delta,n)
=\frac{1}{p_{j}^*(n)}\Bigl\{
&\;A_{ij}\bigl\{H_j^{DATT*}(X_i,n)-\tau_{j}^{DATT*}(\delta,n)\bigr\}\nonumber\\
&+A_{ij}\,\frac{\pi_\delta^*(A_{i(-j)}\mid X_i,n)}{e_{1,j}^*(A_{i(-j)}\mid X_i,n)}\bigl\{\Delta Y_{ij}-m_{1,j}^*(A_{i(-j)},X_i,n)\bigr\}\nonumber\\
&-(1-A_{ij})\,\frac{q_{j}^*(X_i,n)}{1-q_{j}^*(X_i,n)}
\frac{\pi_\delta^*(A_{i(-j)}\mid X_i,n)}
{e_{0,j}^*(A_{i(-j)}\mid X_i,n)}\nonumber\\
&\qquad{}\times\bigl\{\Delta Y_{ij}
-m_{0,j}^*(A_{i(-j)},X_i,n)\bigr\}
+G_{j}^{DATT*}(O_i;\delta,n)\Bigr\},
\label{eq:phiDATT}
\end{align}
\begin{align}
\varphi_{j}^{SATT*}(O_i;\delta,n)
=\frac{1}{p_{j}^*(n)}\Bigl\{
&\;A_{ij}\bigl\{H_j^{SATT*}(X_i,n)-\tau_{j}^{SATT*}(\delta,n)\bigr\}\nonumber\\
&+(1-A_{ij})\,\frac{q_{j}^*(X_i,n)}{1-q_{j}^*(X_i,n)}\,\frac{\pi_\delta^*(A_{i(-j)}\mid X_i,n)}{e_{0,j}^*(A_{i(-j)}\mid X_i,n)}\bigl\{\Delta Y_{ij}-m_{0,j}^*(A_{i(-j)},X_i,n)\bigr\}\nonumber\\
&-(1-A_{ij})\,\one\{A_{i(-j)}=\mathbf 0_{-j}\}\,
\frac{q_{j}^*(X_i,n)}{1-q_{j}^*(X_i,n)}
\frac{1}{e_{0,j}^*(\mathbf 0_{-j}\mid X_i,n)}\nonumber\\
&\qquad{}\times\bigl\{\Delta Y_{ij}
-m_{0,j}^*(\mathbf 0_{-j},X_i,n)\bigr\}
+G_{j}^{SATT*}(O_i;\delta,n)\Bigr\}.
\label{eq:phiSATT}
\end{align}
Each $\varphi_{j}^{\bullet*}$ has size-conditional mean zero. The semiparametric efficiency bound for $\tau^{\bullet}(\delta)$ under $\mathcal{M}_{NP}$ is $\E[\{\varphi^{\bullet*}(O_i;\delta)\}^{2}]$.
\end{thm}

The efficient influence functions in Theorem~\ref{thm:eif} may be viewed as
extensions of the efficient influence functions for the direct and
spillover effects of \citet{lee2025efficient} to the DID design. In particular, the outcome models appearing in their efficient influence functions are replaced by outcome evolution models under the no-anticipation and conditional parallel trends assumptions. Meanwhile, the treated probability, denoted by $p_j^*(n)=\Prob(A_{ij}=1\mid N_i=n)$, and its corresponding propensity score $q_j^*(X_i,n)$ are modified to accommodate the ATT-type targets considered in this paper. There are two main differences between the ATT-type EIF studied in this paper and the ATE-type EIF derived in Corollary 1 of \citet{lee2025efficient}.

First, the leading term
$A_{ij}\bigl\{H_j^{\bullet*}(X_i,n)-\tau_j^{\bullet*}(\delta,n)\bigr\}$ and the
policy-estimation term $G_j^{\bullet*}(O_i;\delta,n)$ of
\eqref{eq:phiDATT}--\eqref{eq:phiSATT} are the counterparts of the first term
of Corollary~1 of \citet{lee2025efficient}, in which the policy weight and
its own influence function appear inside a single sum against a common 
outcome regression. However, such a decomposition is not available in our setting because the two components are weighted differently: the policy average $H_j^{\bullet*}$
enters weighted by the treatment indicator $A_{ij}$, whereas
$G_j^{\bullet*}$ enters weighted by the propensity score $q_j^*(X_i,n)$.
A second difference concerns the size-conditional average
$\bar\tau^{\bullet*}(\delta,N_i)$ in
\eqref{eq:eif}, which has no explicit counterpart in the efficient influence functions of
\citet{lee2025efficient}; again the reason is the treated normalization.
Aggregating \eqref{eq:phiDATT}--\eqref{eq:phiSATT} over units and adding the
size-conditional average in \eqref{eq:eif} merges the two occurrences of
$\tau_j^{\bullet*}(\delta,n)$ into
\[
\tau_j^{\bullet*}(\delta,n)\Bigl\{1-\frac{A_{ij}}{p_j^*(n)}\Bigr\}
=-\frac{\tau_j^{\bullet*}(\delta,n)}{p_j^*(n)}
\bigl\{A_{ij}-p_j^*(n)\bigr\},
\]
which is exactly the contribution of the denominator $p_j^*(n)$.  If the target is an ATE rather than an ATT, the ratio $A_{ij}/p_j^*(n)$ would be replaced by the constant one. In that case, this
contribution would disappear, and $\bar\tau^{\bullet*}(\delta,N_i)$ would cancel
the centering term already present in the unit-level influence functions, leaving
only the global constant $-\tau^{\bullet*}(\delta)$ as in Corollary~1 of
\citet{lee2025efficient}.

If $N_i\equiv1$, then $\mathcal{A}(0)=\{\emptyset\}$, $\pi_\delta^*(\emptyset\mid x,1)=e_{a,1}^*(\emptyset\mid x,1)=1$, the policy-estimation term $G_{j}^{\bullet*}$ has no peer effects and therefore vanishes, $p_{1}^*(1)=\Prob(A_{i1}=1)$, and $\Delta m_{1}^*(\emptyset,x,1)=m_{1,1}^*(x,1)-m_{0,1}^*(x,1)$. Consequently, Equation~\eqref{eq:phiDATT} reduces to
\[
\frac{1}{\Prob(A_{i1}=1)}\Bigl\{A_{i1}\bigl(m_{1,1}^*-m_{0,1}^*-\tau_{1}^{DATT*}\bigr)
+A_{i1}\bigl(\Delta Y_{i1}-m_{1,1}^*\bigr)
-(1-A_{i1})\tfrac{q_{1}^*}{1-q_{1}^*}\bigl(\Delta Y_{i1}-m_{0,1}^*\bigr)\Bigr\},
\]
which coincides with the efficient influence function for the ATT derived in \citet{sant2020doubly}, and $\varphi^{SATT*}\equiv0$. The proof of
Theorem~\ref{thm:eif} is provided in~\ref{app:eif}. In addition, the policy-estimation terms $G_{j}^{\bullet*}$ in
\eqref{eq:Gdatt}--\eqref{eq:Gsatt} vanish when the policy weight is fixed such as under a known
type-B policy.

\begin{rem}[Efficiency under the type-B policy]\label{rem:typeB}
Replace the covariate-specific CIPS weight
$\pi_\delta(a_{-j}\mid X_i,n)$ by the fixed type-B weight, for
$\alpha\in[0,1]$,
\[
\pi_B(a_{-j}\mid n;\alpha)
=\prod_{j'\neq j}\alpha^{a_{j'}}(1-\alpha)^{1-a_{j'}}.
\]
Similar to Theorem~\ref{thm:eif}, for $\bullet\in\{DATT,SATT\}$, define
\[
\begin{aligned}
H_{B,j}^{DATT*}(x,n)
&=\sum_{a_{-j}}\pi_B(a_{-j}\mid n;\alpha)
\Delta m_j^*(a_{-j},x,n),\\
H_{B,j}^{SATT*}(x,n)
&=\sum_{a_{-j}}\pi_B(a_{-j}\mid n;\alpha)
\{m_{0,j}^*(a_{-j},x,n)-m_{0,j}^*(\mathbf 0_{-j},x,n)\},\\
\tau_{B,j}^{\bullet*}(\alpha,n)
&=\frac{\E_n[q_j^*(X_i,n)H_{B,j}^{\bullet*}(X_i,n)]}
{p_j^*(n)},\\
\bar\tau_B^{\bullet*}(\alpha,n)
&=\frac1n\sum_{j=1}^n\tau_{B,j}^{\bullet*}(\alpha,n),\qquad
\tau_B^{\bullet*}(\alpha)
=\sum_{n\le n_{\max}}\nu_n\bar\tau_B^{\bullet*}(\alpha,n).
\end{aligned}
\]
Because $\pi_B$ does not depend on covariates, we have $\partial\pi_B/\partial q_{j'}=0$, implying that there is  no policy-estimation term.  The unit-level efficient influence functions consequently become
\begin{align}
\varphi_{B,j}^{DATT*}(O_i;\alpha,n)
={}&\frac{1}{p_j^*(n)}\Bigl\{
A_{ij}\{H_{B,j}^{DATT*}(X_i,n)
-\tau_{B,j}^{DATT*}(\alpha,n)\}\nonumber\\
&+A_{ij}\frac{\pi_B(A_{i(-j)}\mid n;\alpha)}
{e_{1,j}^*(A_{i(-j)}\mid X_i,n)}
\{\Delta Y_{ij}-m_{1,j}^*(A_{i(-j)},X_i,n)\}\nonumber\\
&-(1-A_{ij})\frac{q_j^*(X_i,n)}{1-q_j^*(X_i,n)}
\frac{\pi_B(A_{i(-j)}\mid n;\alpha)}
{e_{0,j}^*(A_{i(-j)}\mid X_i,n)}
\{\Delta Y_{ij}-m_{0,j}^*(A_{i(-j)},X_i,n)\}
\Bigr\},
\label{eq:phiBDATT}
\end{align}
and
\begin{align}
\varphi_{B,j}^{SATT*}(O_i;\alpha,n)
={}&\frac{1}{p_j^*(n)}\Bigl\{
A_{ij}\{H_{B,j}^{SATT*}(X_i,n)
-\tau_{B,j}^{SATT*}(\alpha,n)\}\nonumber\\
&+(1-A_{ij})\frac{q_j^*(X_i,n)}{1-q_j^*(X_i,n)}
\frac{\pi_B(A_{i(-j)}\mid n;\alpha)}
{e_{0,j}^*(A_{i(-j)}\mid X_i,n)}
\{\Delta Y_{ij}-m_{0,j}^*(A_{i(-j)},X_i,n)\}\nonumber\\
&-(1-A_{ij})\one\{A_{i(-j)}=\mathbf 0_{-j}\}
\frac{q_j^*(X_i,n)}{1-q_j^*(X_i,n)}
\frac{\Delta Y_{ij}-m_{0,j}^*(\mathbf 0_{-j},X_i,n)}
{e_{0,j}^*(\mathbf 0_{-j}\mid X_i,n)}
\Bigr\}.
\label{eq:phiBSATT}
\end{align}
The global efficient influence function is
\begin{equation}\label{eq:eifB}
\varphi_B^{\bullet*}(O_i;\alpha)
=\frac1{N_i}\sum_{j=1}^{N_i}
\varphi_{B,j}^{\bullet*}(O_i;\alpha,N_i)
+\bar\tau_B^{\bullet*}(\alpha,N_i)
-\tau_B^{\bullet*}(\alpha).
\end{equation}
Its variance is the semiparametric efficiency bound for the corresponding
type-B estimand under $\mathcal M_{NP}$.  The type-B efficient influence functions
\eqref{eq:phiBDATT}--\eqref{eq:eifB} are the ATT counterpart of
Theorem~1 of \citet{park2022efficient}.  The second term of their EIF is
exactly $\bar\tau_B^{\bullet*}(\alpha,N_i)-\tau_B^{\bullet*}(\alpha)$ in
\eqref{eq:eifB}. A proof is given in
\ref{app:typeB}.
\end{rem}

\subsection{Estimation and asymptotic properties}
\label{sec:estimation}

Having characterized the efficient influence functions
and their type-B counterparts, we now translate those representations into
feasible cross-fitted estimators and inference procedures.
Collect the nuisance functions as
$\eta=(q_j,e_{a,j},m_{a,j},p_j,\tau_j^\bullet)$.
Under \eqref{eq:assignCI}, $e_{a,j}$ is the derived product displayed above
and therefore need not be fitted separately; we retain the notation to show
the policy-to-observed peer-density ratios and to allow either plug-in
construction from $\widehat q$ or separate consistent estimation.
The shifted propensities $q_{j,\delta}$, policy probabilities $\pi_\delta$,
policy derivatives $g_{j,a_{-j},j'}^{\prime}$, and terms $G_j^\bullet$ are explicit
functionals of $(q_j,m_{a,j})$.  For a generic nuisance value $\eta$, let
$\varphi_j^\bullet(O_i;\delta,n,\eta)$ denote
\eqref{eq:phiDATT} or \eqref{eq:phiSATT} with every starred quantity replaced
by its component of $\eta$, and define the uncentered efficient influence
function
\begin{equation}\label{eq:momentfn}
\begin{aligned}
\Gamma^\bullet(O_i;\delta,\eta)
&=\varphi^\bullet(O_i;\delta,\eta)
+\tau^\bullet(\delta;\eta)\\
&=\frac1{N_i}\sum_{j=1}^{N_i}
\left\{\tau_j^\bullet(\delta,N_i;\eta)
+\varphi_j^\bullet(O_i;\delta,N_i,\eta)\right\}.
\end{aligned}
\end{equation}
At the truth,
$\E[\Gamma^\bullet(O_i;\delta,\eta^*)]
=\tau^\bullet(\delta)$ by Theorem~\ref{thm:eif}.
Randomly partition the $M$ clusters into a fixed number $K$ of folds
$\mathcal I_1,\ldots,\mathcal I_K$.  For each $k$, use the full training
complement $\mathcal I_k^c$ both to fit the regression nuisances
$(q_j,e_{a,j},m_{a,j})$ and to form the derived quantities below.  Thus all
components of $\widehat\eta^{(-k)}$ are functions of $\mathcal I_k^c$ and
are independent of observations in the fold $\mathcal I_k$.
Writing $\widehat\E_n^{(-k)}$ for the empirical mean over all size-$n$ clusters in
$\mathcal I_k^c$, define $\widehat p_j^{(-k)}(n)
=\widehat\E_n^{(-k)}[A_{ij}]$ and
\begin{equation}\label{eq:derivednuis}
\begin{aligned}
\widehat{\tau}_j^{DATT,(-k)}(\delta,n)
&=\frac{\widehat\E_n^{(-k)}\!\left[
\widehat q_j^{(-k)}(X_i,n)\sum_{a_{-j}}\widehat\pi_\delta^{(-k)}(a_{-j}\mid X_i,n)
\{\widehat m_{1,j}^{(-k)}(a_{-j},X_i,n)-\widehat m_{0,j}^{(-k)}(a_{-j},X_i,n)\}\right]}
{\widehat p_j^{(-k)}(n)},\\
\widehat{\tau}_j^{SATT,(-k)}(\delta,n)
&=\frac{\widehat\E_n^{(-k)}\!\left[
\widehat q_j^{(-k)}(X_i,n)\sum_{a_{-j}}\widehat\pi_\delta^{(-k)}(a_{-j}\mid X_i,n)
\{\widehat m_{0,j}^{(-k)}(a_{-j},X_i,n)-\widehat m_{0,j}^{(-k)}(\mathbf 0_{-j},X_i,n)\}\right]}
{\widehat p_j^{(-k)}(n)}.
\end{aligned}
\end{equation}

In implementation, the probability estimates may be truncated to the
intervals specified in Assumption~\ref{asm:2}(i).  The cross-fitted one-step
estimator is
\begin{equation}\label{eq:tauhat}
\widehat{\tau}^{\bullet}(\delta)
=\frac1M\sum_{k=1}^K\sum_{i\in\mathcal I_k}
\Gamma^\bullet\bigl(O_i;\delta,\widehat\eta^{(-k)}\bigr),
\qquad \bullet\in\{DATT,SATT\}.
\end{equation}

For a function of $X_i$ within size-$n$ clusters, write
$\|f\|_{2,n}=\{\E_n[f(X_i)^2]\}^{1/2}$ and set
\[
\varepsilon_q=\max_{k,j,n}\|\widehat q_j^{(-k)}-q_j^*\|_{2,n},\qquad
\varepsilon_e=\max_{k,a,j,a_{-j},n}
\|\widehat e_{a,j}^{(-k)}(a_{-j}\mid\cdot,n)-e_{a,j}^*(a_{-j}\mid\cdot,n)\|_{2,n},
\]
\[
\varepsilon_m=\max_{k,a,j,a_{-j},n}
\|\widehat m_{a,j}^{(-k)}(a_{-j},\cdot,n)-m_{a,j}^*(a_{-j},\cdot,n)\|_{2,n}.
\]
We next impose conditions on the nuisance estimators that are sufficient for
asymptotic normality and efficiency of the cross-fitted estimator. 

\begin{asm}[Nuisance estimation]\label{asm:2}
For every fold, the following conditions hold.
\begin{enumerate}[label=(\roman*)]
\item (Boundedness) There are constants $c'>0$ and $C<\infty$ such that
$c'\le\widehat q_j^{(-k)}\le1-c'$,
$\widehat e_{a,j}^{(-k)}\ge c'$,
$\widehat p_j^{(-k)}(n)\ge c'$, $|\widehat m_{a,j}^{(-k)}|\le C$, and
$|\widehat{\tau}_j^{\bullet,(-k)}(\delta,n)|\le C$ almost surely.  (The
corresponding population bound on
$\E[\Delta Y_{ij}^2\mid A_i,X_i,N_i]$ is part of
Assumption~\ref{asm:1}(v).)
\item (Consistency)
$\varepsilon_q+\varepsilon_e+\varepsilon_m=o_p(1)$.
\item (Nuisance convergence) The regression nuisances satisfy
$\varepsilon_q^2+(\varepsilon_q+\varepsilon_e)\varepsilon_m
=o_p(M^{-1/2})$,
and, uniformly over $k$, $\bullet\in\{DATT,SATT\}$, $j$, and $n$, the
difference between the empirical and the population average of the fitted
numerator integrand in \eqref{eq:derivednuis} is $o_p(1)$.
\end{enumerate}
\end{asm}
Assumption~\ref{asm:2} can be viewed as the ATT counterpart of
conditions (B1)--(B8) of \citet{lee2025efficient}. The first part of condition (iii) holds, for example, if all regression
nuisances converge faster than $M^{-1/4}$.  The isolated $\varepsilon_q^2$
term is generated by estimating the CIPS policy and cannot be traded against
outcome-regression accuracy; by contrast, the remaining products allow the
usual assignment--outcome rate trade-off.  The second part is a training-sample stability condition. Together with condition (ii), it implies the consistency of each $\widehat{\tau}_j^{\bullet,(-k)}$.

\begin{thm}[Asymptotic normality and local efficiency]\label{thm:asymp}
Suppose Assumptions~\ref{asm:1} and~\ref{asm:2} hold and $K$ is fixed.  Then,
for each fixed $\delta\in(0,\infty)$ and
$\bullet\in\{DATT,SATT\}$,
\[
\sqrt M\{\widehat{\tau}^{\bullet}(\delta)
-\tau^{\bullet}(\delta)\}
=\frac1{\sqrt M}\sum_{i=1}^M
\varphi^{\bullet*}(O_i;\delta)+o_p(1)
\xrightarrow{d}\mathcal N(0,\sigma_\bullet^2),
\]
where
$\sigma_\bullet^2
=\E[\{\varphi^{\bullet*}(O_i;\delta)\}^2]$, assumed positive.
Thus the estimator attains the efficiency bound in Theorem~\ref{thm:eif}.
Moreover,
\begin{equation}\label{eq:Vhat}
\widehat V^{\bullet}(\delta)
=\frac1M\sum_{k=1}^K\sum_{i\in\mathcal I_k}
\left\{\Gamma^\bullet(O_i;\delta,\widehat\eta^{(-k)})
-\widehat{\tau}^{\bullet}(\delta)\right\}^2
\end{equation}
satisfies
$\widehat V^{\bullet}(\delta)\to_p\sigma_\bullet^2$.
Consequently,
\[
\widehat{\tau}^{\bullet}(\delta)
\ \pm\ z_{1-\gamma/2}
\sqrt{\widehat V^{\bullet}(\delta)/M}
\]
is an asymptotically valid $(1-\gamma)$ confidence interval.
\end{thm}
Theorem~\ref{thm:asymp} establishes that the
cross-fitted CIPS estimator is asymptotically linear, attains the
semiparametric efficiency bound, and admits consistent variance
estimation.  A proof is provided in~\ref{app:asymp}. We next contrast these properties with the robustness
available under the fixed type-B policy.

\begin{rem}[Double robustness under the type-B policy]\label{rem:typeBdr}
Because the type-B policy weight $\pi_B$ is fixed rather than a functional
of the observed treatment law,
$\partial\pi_B/\partial q_{j'}=0$ and the additional policy term
$G_j^\bullet$ is absent.  Consequently, its unit-level efficient influence
function in Remark~\ref{rem:typeB} has the augmented
inverse-probability-weighted (AIPW) form
\[
\varphi_{B,j}^{\bullet*}(O_i;\alpha,n)
=\frac{1}{p_j^*(n)}
\left[
A_{ij}\{H_{B,j}^{\bullet*}(X_i,n)
-\tau_{B,j}^{\bullet*}(\alpha,n)\}
+R_{B,Y,j}^{\bullet}(O_i)
\right],
\]
where $R_{B,Y,j}^{\bullet}$ is an outcome residual correction constructed using inverse probability weighting, given explicitly in
\eqref{eq:phiBDATT}--\eqref{eq:phiBSATT}.  The corresponding conditional
bias, apart from the asymptotically negligible empirical normalization
term, therefore contains only products of assignment errors and
outcome regression errors.  If the outcome regressions are correctly specified, the weighted residual corrections have conditional mean zero; if the assignment nuisance functions $(q_j,e_{a,j})$ are correctly specified, inverse-probability weighting cancels errors in the outcome-regression component. Hence, the type-B estimator is consistent if either nuisance family is correctly specified and is therefore doubly robust. The type-B DID estimators for the DATT and SATT, together with their asymptotic properties, are provided in~\ref{app:typeB}.
\end{rem}

\section{Simulation}
\label{sec:simulation}

We conduct a Monte Carlo study to assess the finite-sample performance of
the CIPS estimators in \eqref{eq:tauhat} and the type-B estimators presented in the Appendix.  The design reports all four combinations of policy and
causal contrast: CIPS DATT, CIPS SATT, type-B DATT, and type-B SATT. We generate $M\in\{1000,2000\}$ independent clusters.  Cluster size
$N_i\in\{1,2,3\}$ is drawn from
$\nu=(\nu_1,\nu_2,\nu_3)=(0.30,0.40,0.30)$, and the member covariates are
independent, with $X_{ij}\sim\mathcal N(0,1)$.  Let
$Z_i=\sum_{k=1}^{N_i}X_{ik}$.  Conditional on $(X_i,N_i)$, treatments are
independent within a cluster and
\begin{equation*}
 A_{ij}\mid X_i,N_i=n
 \sim\operatorname{Bernoulli}\{q_j(Z_i,n)\},\qquad
 q_j(z,n)=\operatorname{expit}
 \{a_0+a_N(n-1)+a_J(j-1)+a_Xz\},
\end{equation*}
with $(a_0,a_N,a_J,a_X)=(-0.20,0.10,0.15,0.35)$.  This assignment
mechanism satisfies the conditional independence restriction
\eqref{eq:assignCI}, while the common dependence on $Z_i$ allows the
members' treatments to remain marginally dependent.
Define the position-specific spillover coefficient
$b_{s,k}=b_{s0}+b_{s1}(k-1)$.  The observed outcome evolution is generated as
\begin{equation*}
\Delta Y_{ij}
=b_0+b_AA_{ij}
+\sum_{k\ne j}b_{s,k}A_{ik}
+gZ_i+b_J(j-1)+b_N(N_i-1)+\varepsilon_{ij},
\qquad
\varepsilon_{ij}\sim\mathcal N(0,\sigma^2),
\end{equation*}
where
$(b_0,b_A,b_{s0},b_{s1},g,b_J,b_N,\sigma)
=(0,0.30,0.10,0.04,0.30,0.05,0.05,1)$.
 The conditional DATT and SATT are
$\theta_j^{DATT}(a_{-j},x,n)=b_A, \ 
\theta_j^{SATT}(a_{-j},x,n)
=\sum_{k\ne j}b_{s,k}a_k$.
Let $f_n$ denote the density of $\mathcal N(0,n)$,
$q_{k,\delta}(z,n)=\delta q_k(z,n)/
\{\delta q_k(z,n)+1-q_k(z,n)\}$, and
$p_j(n)=\int q_j(z,n)f_n(z)\,dz$,  
$r_{jk,\delta}(n)
=\int q_j(z,n)q_{k,\delta}(z,n)f_n(z)\,dz / p_j(n)$.
The factor $r_{jk,\delta}(n)$ is the CIPS treatment probability of peer
$k$, averaged over the covariate distribution of treated focal units
$j$.  In particular, it retains the pointwise product
$q_j(X_i,n)q_{k,\delta}(X_i,n)$ required by
\eqref{eq:idtauSATT}.  The four super-population truths are
\begin{align*}
\tau_{\mathrm{CIPS}}^{DATT}(\delta)
&=b_A,\
\tau_{\mathrm{CIPS}}^{SATT}(\delta)
=\sum_n\nu_n\frac1n\sum_{j=1}^n
  \sum_{k\ne j}b_{s,k}r_{jk,\delta}(n),\label{eq:sim-truth-cips}\\
\tau_B^{DATT}(\alpha)
&=b_A,\
\tau_B^{SATT}(\alpha)
=\sum_n\nu_n\frac1n\sum_{j=1}^n
  \sum_{k\ne j}b_{s,k}\alpha.
\end{align*}
We use $\delta\in\{0.5,1,2\}$ for CIPS and
$\alpha\in\{0.25,0.5,0.75\}$ for type-B.
The regression nuisances are estimated by correctly specified parametric
working models.  The propensity model is logistic in the ordered unit
label, cluster size, and complete ordered cluster covariate vector; the
outcome model is linear in the focal treatment, complete ordered peer
treatment vector, ordered unit label, cluster size, and complete ordered
cluster covariate vector.  Both models are fitted within $K=5$
cluster-level cross-fitting folds.
Estimation uses base \textsf{R} only, with \texttt{glm} and \texttt{lm} and no
data-adaptive components.  The estimated probabilities are truncated at $c'=0.05$ following Assumption~\ref{asm:2}(i). Specifically, $\widehat q_j^{(-k)}$ is clipped to the interval $[0.05,0.95]$, while the peer densities $\widehat e_{a,j}^{(-k)}$ and the treated marginal probabilities $\widehat p_j^{(-k)}(n)$ are bounded below by $0.05$. No clusters or units are removed, so the estimand remains unchanged.
Across the $D=2000$
replications, we calculate the bias, root mean squared error (RMSE),
empirical standard deviation (SD), average analytic standard error
($\mathrm{SE}$), and coverage of the nominal $95\%$ Wald confidence
interval.

\begin{table}[ht]
\centering
\small
\setlength{\tabcolsep}{4.2pt}
\caption{Monte Carlo results for the CIPS and type-B DATT and SATT estimators.}
\label{tab:sim-new}
\begin{tabular}{cc ccccc ccccc}
\toprule
 & & \multicolumn{5}{c}{$M=1000$} & \multicolumn{5}{c}{$M=2000$} \\
\cmidrule(lr){3-7}\cmidrule(lr){8-12}
Policy parameter & Truth & Bias & RMSE & SD & $\mathrm{SE}$ & Cov.
& Bias & RMSE & SD & $\mathrm{SE}$ & Cov. \\
\midrule
\multicolumn{12}{l}{\emph{Panel A. CIPS direct effect}}\\
\addlinespace
$0.50$ & $0.300$ & $\phantom{-}0.000$ & $0.055$ & $0.055$ & $0.055$ & $0.950$
& $\phantom{-}0.000$ & $0.039$ & $0.039$ & $0.038$ & $0.948$ \\
$1.00$ & $0.300$ & $\phantom{-}0.000$ & $0.053$ & $0.053$ & $0.052$ & $0.949$
& $\phantom{-}0.000$ & $0.037$ & $0.037$ & $0.037$ & $0.948$ \\
$2.00$ & $0.300$ & $-0.001$ & $0.055$ & $0.055$ & $0.054$ & $0.948$
& $-0.001$ & $0.038$ & $0.038$ & $0.038$ & $0.949$ \\
\addlinespace
\multicolumn{12}{l}{\emph{Panel B. CIPS spillover effect}}\\
\addlinespace
$0.50$ & $0.053$ & $-0.001$ & $0.055$ & $0.055$ & $0.047$ & $0.956$
& $\phantom{-}0.001$ & $0.037$ & $0.037$ & $0.034$ & $0.953$ \\
$1.00$ & $0.074$ & $-0.001$ & $0.063$ & $0.063$ & $0.057$ & $0.949$
& $\phantom{-}0.001$ & $0.043$ & $0.043$ & $0.040$ & $0.954$ \\
$2.00$ & $0.093$ & $\phantom{-}0.000$ & $0.070$ & $0.070$ & $0.064$ & $0.949$
& $\phantom{-}0.001$ & $0.048$ & $0.048$ & $0.045$ & $0.952$ \\
\addlinespace
\multicolumn{12}{l}{\emph{Panel C. Type-B direct effect}}\\
\addlinespace
$0.25$ & $0.300$ & $\phantom{-}0.000$ & $0.066$ & $0.066$ & $0.066$ & $0.948$
& $\phantom{-}0.000$ & $0.047$ & $0.048$ & $0.046$ & $0.943$ \\
$0.50$ & $0.300$ & $\phantom{-}0.000$ & $0.056$ & $0.056$ & $0.055$ & $0.952$
& $\phantom{-}0.000$ & $0.039$ & $0.039$ & $0.039$ & $0.950$ \\
$0.75$ & $0.300$ & $-0.001$ & $0.057$ & $0.057$ & $0.057$ & $0.949$
& $-0.001$ & $0.040$ & $0.040$ & $0.040$ & $0.947$ \\
\addlinespace
\multicolumn{12}{l}{\emph{Panel D. Type-B spillover effect}}\\
\addlinespace
$0.25$ & $0.033$ & $-0.001$ & $0.031$ & $0.031$ & $0.028$ & $0.952$
& $\phantom{-}0.000$ & $0.021$ & $0.021$ & $0.020$ & $0.960$ \\
$0.50$ & $0.066$ & $-0.001$ & $0.054$ & $0.054$ & $0.049$ & $0.949$
& $\phantom{-}0.001$ & $0.036$ & $0.036$ & $0.034$ & $0.954$ \\
$0.75$ & $0.099$ & $\phantom{-}0.000$ & $0.069$ & $0.069$ & $0.064$ & $0.946$
& $\phantom{-}0.001$ & $0.047$ & $0.047$ & $0.045$ & $0.953$ \\
\bottomrule
\end{tabular}

\vspace{2pt}
\begin{minipage}{0.96\textwidth}
\footnotesize
\emph{Notes:} Bias is the Monte Carlo mean of
$\widehat\tau-\tau$; RMSE is its root mean squared error; SD is the
empirical standard deviation; $\mathrm{SE}$ is the average analytic
standard error; and Cov.\ denotes the empirical coverage of nominal $95\%$ Wald intervals. Results use
$2000$ replications and $K=5$ cluster-level cross-fitting folds.
\end{minipage}
\end{table}

Three features of Table~\ref{tab:sim-new} are noteworthy.  First, all four
estimators have small finite-sample bias: the absolute bias is at
most $0.0012$ over all policy parameters and sample sizes.  Second, RMSE and SD
decrease substantially when $M$ rises from $1000$ to $2000$, consistent
with the root-$M$ behavior in Theorem~\ref{thm:asymp}.  For the
direct effects, SD and SE agree closely at
both sample sizes.  For the spillover
effects, the average standard error is somewhat smaller than SD at $M=1000$,
with ratios between $0.87$ and $0.92$; the gap narrows to between $0.91$ and
$0.95$ at $M=2000$, as the rare complete peer configurations become better
represented.  Third, the empirical coverage of all four estimators ranges from $0.943$ to
$0.960$ and is therefore close to the nominal $95\%$.
Together, these results support the asymptotic linearity and variance
conclusions for both the estimated CIPS policy, whose EIF contains
$G_j^\bullet$, and the fixed type-B policy, whose EIF omits that term, as presented
in Remark~\ref{rem:typeBdr} and~\ref{app:typeB}.

\section{Application: social pensions and intra-household spillovers}
\label{sec:application}

We illustrate the estimator in the setting of \citet{huang2021power}, who study
China's New Rural Pension Scheme (NRPS), a large noncontributory social pension
that was rolled out county by county from September $2009$ and was nearly universal by the end
of $2012$.  Once a county was covered, every rural resident aged $16$ and over
could enroll voluntarily, and enrollees aged $60$ and over received a basic
pension of $55$ yuan per month regardless of previous earnings or income. 
\citet{huang2021power} regress the outcome on the county coverage indicator
$NRPS_{ct}$ with county and year fixed effects, and report, among the
age-eligible rural elderly, a reduction in farmwork participation of $3.6$
percentage points.

This setting is a natural application of the theory in
Sections~\ref{sec:setup}--\ref{sec:efficiency} because
\citet{huang2021power} mention the possibility of spillover effects, while their
main analysis reports an intention-to-treat parameter defined by county
coverage.  We take that broader spillover concern as motivation for a
complementary household-level decomposition, without requiring the county-level
estimand to identify either component separately.  Treating the household as
the cluster, $\tau^{DATT}(\delta)$ is the effect of a member's own pension
receipt on that member, and $\tau^{SATT}(\delta)$ is the effect on a member of a
co-resident's receipt with the member's own receipt held at non-receipt.

Following \citet{huang2021power}, we use the China Health and Retirement
Longitudinal Study (CHARLS), waves $2011$ (pre) and $2013$ (post), so that
treatment is realized only in the second period.  To align the application with the no-anticipation condition, we drop $40$ of the $150$ CHARLS counties whose rollout year is $2010$ or earlier, since they are covered in both waves; the $110$ retained counties are untreated at the $2011$ baseline.  Each household is a cluster ($i=1,\dots,M$) and
each co-resident rural-\emph{hukou} member aged $60$ and over is a unit
($j=1,\dots,N_i$), with households drawn i.i.d.\ and within-household
interference unrestricted.
The age restriction is imposed because the basic pension is available only to individuals aged $60$ and above; therefore, members below this age are not eligible to receive the treatment. This restriction is also adopted in \citet{huang2021power}. The
treatment $A_{ij}=1$ indicates that member $j$ is receiving an NRPS pension by
$2013$, and the outcome is participation in farmwork, the dependent variable in
column~$2$ of Table~$2$ in \citet{huang2021power}. The analysis sample has
$3{,}217$ members in $2{,}335$ households, of whom $69.5\%$ are receiving a
pension by $2013$; $882$ households contribute two members and the remainder one,
so the cluster size $N_i$ is a bounded random variable entering \eqref{eq:tauhat}
directly and the size law $\nu_n$ is its empirical distribution.

The covariate vector $X_i$ stacks the baseline ($2011$) covariates of every
household member: gender, age, age$^2$, and dummies for primary school and for
junior high or above, together with an indicator for a low-income county that controls for systematic county-level differences relevant to pension take-up and farmwork.
The nuisance functions $q_j$ and $m_{a,j}$ are estimated via a super learner \citep{vanderlaan2007super}, which constructs an ensemble estimator from a library of parametric and data-adaptive algorithms.  The library consists of \texttt{SL.glm},
\texttt{SL.glmnet}, \texttt{SL.earth}, \texttt{SL.gam}, \texttt{SL.xgboost},
\texttt{SL.ranger} and \texttt{SL.nnet}, together with the marginal-mean learner
\texttt{SL.mean}, and is implemented with the \texttt{SuperLearner} package for
\textsf{R} \citep{polley2024superlearner}.  Ensemble weights are chosen by five-fold internal cross-validation,
and both the internal folds and the $K=5$ outer cross-fitting folds are
assigned at the household level, so that no household is split across folds.
Estimated probabilities are truncated at $c'=0.05$ as in
Section~\ref{sec:simulation}.
Figure~\ref{fig:farmwork} reports $\widehat\tau^{DATT}(\delta)$ and
$\widehat\tau^{SATT}(\delta)$ on a grid of $41$ values of $\delta\in[0.5,2]$,
log-spaced about the factual policy $\delta=1$, with $95\%$ confidence intervals
from \eqref{eq:Vhat}.  At $\delta=1$ the direct effect is
$\widehat\tau^{DATT}(1)=-0.061$ ($\mathrm{SE}=0.021$, $95\%$ CI
$[-0.101,\,-0.020]$), and the interval excludes zero at every point of the grid.
A member who receives the pension is thus about $6.1$ percentage points less
likely to do farmwork, against a baseline participation rate of $0.571$ in
$2011$.  This is close to the $-0.036$ ($0.018$) that
\citet{huang2021power} report, and like their estimate, it is significant: receiving the
pension lowers the probability of doing farmwork.  The two coefficients are not
the same estimand.  Theirs is an intention-to-treat effect of county coverage,
which averages the response over the covered but non-enrolled, whereas ours is
the effect of individual receipt on those who receive. 

\begin{figure}[ht]
\centering
\includegraphics[width=\textwidth]{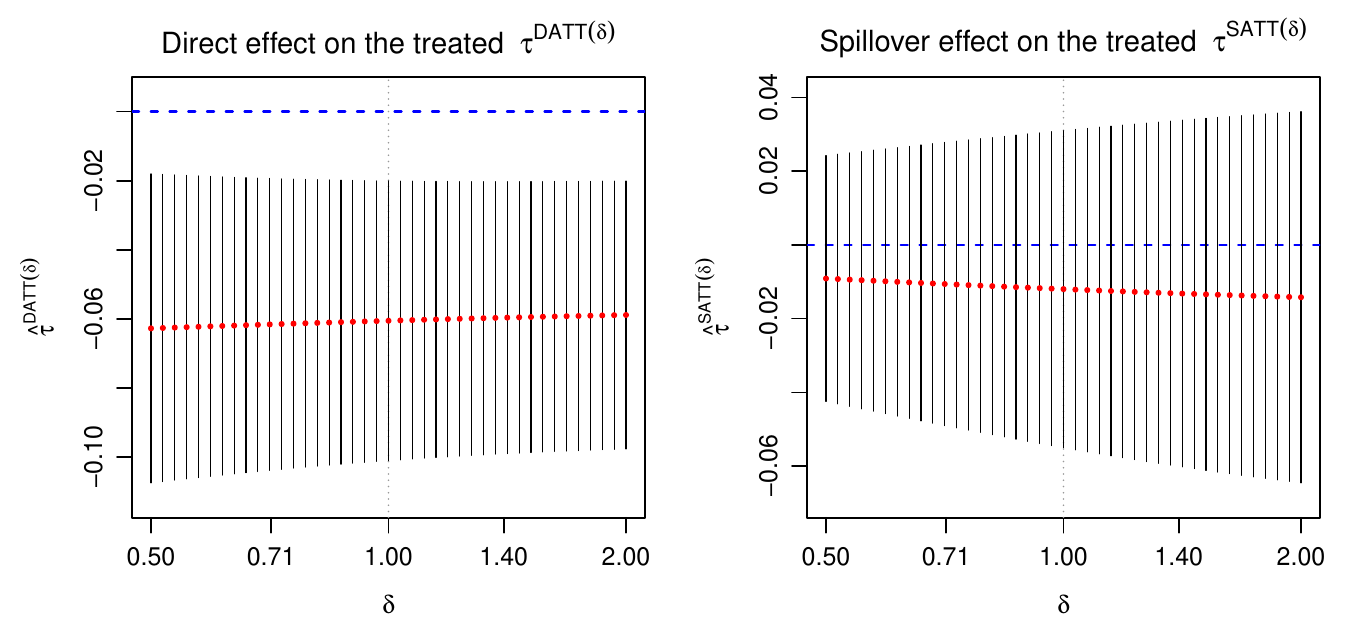}
\caption{Cross-fitted CIPS estimates of the direct
($\tau^{DATT}(\delta)$, left) and spillover ($\tau^{SATT}(\delta)$, right)
effects on \emph{participation in farmwork}, over the CIPS odds multiplier
$\delta$.  Red dots are point estimates and vertical bars are $95\%$ confidence
intervals computed from \eqref{eq:Vhat}; the dashed line marks the null.  The
direct effect is negative and significant at every $\delta$; the spillover effect
is not distinguishable from zero anywhere on the grid.  Note the difference in
vertical scale between the two panels.}
\label{fig:farmwork}
\end{figure}

The spillover effect points the same way but cannot be signed.  At $\delta=1$,
$\widehat\tau^{SATT}(1)=-0.012$ ($\mathrm{SE}=0.022$): among pension recipients,
a co-resident's receipt would lower the probability of doing farmwork by $1.2$
percentage points even if the member's own receipt were withheld.  The interval
contains zero at every point of the grid, so we report the sign as suggestive only.
One possible reason for this imprecision is that estimation of the spillover contrast relies on households whose members have different treatment statuses, and only $88$ of the $882$ two-member households exhibit such assignments.

The shape of the two curves in $\delta$ is informative in a way the levels are
not, and here it resembles the additive specification in
Example~\ref{ex:linear}.  That example shows that when the treatment response is
additive in own and peer treatment,
$\tau^{DATT}(\delta)=\beta^{D}$ does not depend on $\delta$, whereas
$\tau^{SATT}(\delta)=\beta^{S}\,\E[N_i-1]\,q_\delta$ moves with the policy.
In the application, the direct effect moves only from $-0.063$ at
$\delta=0.5$ to $-0.059$ at $\delta=2$.  This near-flatness may reflect a
weak own--peer interaction, as in Example~\ref{ex:linear}.
In sum, the direct effect estimate agrees in sign and is
comparable in magnitude to the county-level intention-to-treat estimate in
\citet{huang2021power}.  We therefore view
the household-level direct and spillover effects as a complementary
decomposition, rather than as a replication of their result or as evidence on
the sign of the spillover.

\section{Discussion}
\label{sec:discussion}
This paper develops identification, efficiency, and inference for
policy-indexed direct and spillover average treatment effects on the treated
in two-period difference-in-differences designs with partial interference.
Under no anticipation and configuration-specific conditional parallel trends,
the CIPS policy shifts the observed covariate-dependent treatment odds and
traces causal effects across counterfactual allocations that preserve the
observed conditional support. We derive the efficient influence functions,
showing how ATT changes the EIF relative to existing cross-sectional ATE-type
representations and how the dependence of the CIPS policy on the observed
propensity scores generates an additional policy-estimation term that
disappears under a fixed type-B policy. The resulting cluster-level
cross-fitted estimators are asymptotically linear, locally efficient, and
equipped with consistent variance estimators under high-level nuisance-rate
conditions. The comparison with the fixed type-B policy also reveals a
substantive robustness distinction: the type-B estimator is doubly robust for
consistency, whereas the CIPS estimator requires consistent estimation of the
assignment mechanism because that mechanism defines the target policy.
Monte Carlo results show small finite-sample bias and coverage close to the
nominal level. In the pension application, the estimated direct effect at
$\delta=1$ is a $6.1$-percentage-point reduction in farmwork, while the
within-household spillover estimate is imprecise and not statistically
distinguishable from zero.

Several limitations point to future work. The partial interference assumption restricts spillovers to within-household pairs, and extending the framework to more general network structures would broaden its applicability. The two-period DID design also abstracts from staggered adoption and time-varying effects, so adapting the CIPS estimand to multi-period settings is a natural next step. Further, the CIPS estimator's one-sided robustness permits outcome-model misspecification when the assignment nuisance functions are consistently estimated, but, unlike double robustness, a correctly specified outcome model cannot compensate for misspecification of the propensity model; whether a doubly robust policy-estimation term exists without sacrificing efficiency remains open. 
Finally, the policy parameter---how far to shift the propensity score from the factual allocation---is treated here as fixed rather than chosen; casting it instead as a decision variable to be optimized against a stated welfare or budget criterion would connect the framework to the policy learning literature and yield a data-driven recommendation rather than an estimate at an analyst-specified shift.


\clearpage
\bibliographystyle{elsarticle-harv}
\bibliography{references}
\clearpage
\appendix

\section{Notation}
\label{app:notation}

To help guide the proofs, we collect the notation used throughout the paper
in the following table. The symbols are listed roughly in their order of
appearance in the main text.

\begingroup
\renewcommand{\arraystretch}{1.3}
\footnotesize
\begin{longtable}{p{0.28\textwidth} p{0.66\textwidth}}
\toprule
\textbf{Symbol} & \textbf{Definition / expression}\\
\midrule
\endfirsthead
\toprule
\textbf{Symbol} & \textbf{Definition / expression}\\
\midrule
\endhead
\bottomrule
\endlastfoot

$M$ & Number of i.i.d.\ clusters; asymptotics are $M\to\infty$ with bounded cluster size\\
$\Prob$ & Super-population distribution from which clusters are drawn\\
$i$ & Cluster index, $i=1,\dots,M$\\
$j$ & Unit index within cluster $i$, $j=1,\dots,N_i$\\
$j'$ & Peer-unit index, with $j'\ne j$ when unit $j$ is the focal unit\\
$k$ & Local index: a unit index in the assignment-score calculations of the proof of Theorem~\ref{thm:eif}, and a fold index in the cross-fitting context\\
$N_i$ & Random size of cluster $i$, with $\Prob(N_i\le n_{\max})=1$\\
$n_{\max}$ & Fixed finite upper bound on cluster size\\
$n$ & Generic cluster size, $n\le n_{\max}$\\
$\mathcal A(s)$ & Set of all length-$s$ binary vectors\\
$t\in\{1,2\}$ & Time period in the main text; treatment is realized only in period $2$. The symbol $t$ is reused only in the proof of Theorem~\ref{thm:eif} as a parametric-submodel index\\
$Y_{ijt}$ & Outcome of unit $j$ in cluster $i$ at time $t$\\
$X_{ij}\in\mathbb R^p$ & Vector of pre-treatment covariates of unit $j$\\
$A_{ij}\equiv A_{ij2}$ & Period-$2$ binary treatment indicator of unit $j$\\
$Y_{it}$ & Stacked outcome vector, $(Y_{i1t},\dots,Y_{iN_it})^\top$\\
$A_i$ & Stacked treatment vector, $(A_{i1},\dots,A_{iN_i})^\top$\\
$X_i$ & Stacked cluster covariate vector, $(X_{i1}^\top,\dots,X_{iN_i}^\top)^\top$\\
$A_{i(-j)}$ & Treatment vector of all units in cluster $i$ except unit $j$\\
$\mathbf 0_{-j}$ & All-zero peer-treatment vector in $\mathcal A(N_i-1)$\\
$O_i$ & Observed cluster data, $(Y_{i1},Y_{i2},A_i,X_i,N_i)$\\
$\Delta Y_{ij}$ & First-differenced outcome, $Y_{ij2}-Y_{ij1}$\\
$W_{ij}$ & Reduced unit-level data, $(\Delta Y_{ij},A_{ij},A_{i(-j)},X_i,N_i)$\\
$\nu_n$ & Cluster-size law, $\Prob(N_i=n)$\\
$\E_n[\cdot]$ & Size-conditional expectation, $\E[\cdot\mid N_i=n]$\\
\midrule

$a_i=(a_{ij},a_{i(-j)})$ & Generic cluster treatment vector, separated into the focal and peer components\\
$a\in\{0,1\}$ & Generic treatment value of the focal unit\\
$a_{-j}$ & Generic peer-treatment configuration in $\mathcal A(n-1)$\\
$Y_{ijt}(a_i)$ & Potential outcome under cluster treatment $a_i$\\
$q_j(x,n)$ & Individual propensity score, $\Prob(A_{ij}=1\mid X_i=x,N_i=n)$\\
$\delta\in(0,\infty)$ & CIPS odds multiplier\\
$q_{j,\delta}(x,n)$ & Shifted propensity, $\dfrac{\delta q_j(x,n)}{\delta q_j(x,n)+1-q_j(x,n)}$\\
$\pi_\delta(a_{-j}\mid x,n)$ & Conditional CIPS peer law,
$\prod_{j'\ne j}q_{j',\delta}(x,n)^{a_{j'}}
\{1-q_{j',\delta}(x,n)\}^{1-a_{j'}}$\\
$\alpha\in[0,1]$ & Common peer-treatment probability under a type-B policy\\
$\pi_B(a_{-j}\mid n;\alpha)$ & Type-B peer law,
$\prod_{j'\ne j}\alpha^{a_{j'}}(1-\alpha)^{1-a_{j'}}$\\
\midrule

$m_{a,j}(a_{-j},x,n)$ & Outcome regression,
$\E[\Delta Y_{ij}\mid A_{ij}=a,A_{i(-j)}=a_{-j},X_i=x,N_i=n]$\\
$\Delta m_j(a_{-j},x,n)$ & Regression contrast,
$m_{1,j}(a_{-j},x,n)-m_{0,j}(a_{-j},x,n)$\\
$e_{a,j}(a_{-j}\mid x,n)$ & Peer-configuration probability,
$\Prob(A_{i(-j)}=a_{-j}\mid A_{ij}=a,X_i=x,N_i=n)$.
Under \eqref{eq:assignCI},
$e_{0,j}(a_{-j}\mid x,n)=e_{1,j}(a_{-j}\mid x,n)
=\prod_{j'\ne j}q_{j'}(x,n)^{a_{j'}}
\{1-q_{j'}(x,n)\}^{1-a_{j'}}$\\
$p_j(n)$ & Treated marginal probability,
$\Prob(A_{ij}=1\mid N_i=n)=\E_n[q_j(X_i,n)]$\\
$c$ & Population overlap constant in Assumption~\ref{asm:1}(iv)\\
$C_Y$ & Almost-sure bound on the conditional second moment
$\E[\Delta Y_{ij}^{2}\mid A_i,X_i,N_i]$ in Assumption~\ref{asm:1}(v)\\
\midrule

$\bullet$ & Effect-type placeholder, $\bullet\in\{DATT,SATT\}$; a symbol
written with $\bullet$ is read with either label in its place\\
$\theta_j^{DATT}(a_{-j},x,n)$ & Covariate-conditional direct effect on the treated in \eqref{eq:thetaDATTx}\\
$\theta_j^{SATT}(a_{-j},x,n)$ & Covariate-conditional spillover effect on the treated in \eqref{eq:thetaSATTx}\\
$\theta_j^{TATT}(a_{-j},x,n)$ & Total effect
$\theta_j^{DATT}(a_{-j},x,n)+\theta_j^{SATT}(a_{-j},x,n)$\\
$\tau_j^\bullet(\delta,n)$ & Unit-level policy effect in \eqref{eq:tauj}\\
$\bar\tau^{\bullet}(\delta,n)$ & Size-conditional average,
$n^{-1}\sum_{j=1}^n\tau_j^\bullet(\delta,n)$\\
$\tau^\bullet(\delta)$ & Aggregate policy effect,
$\sum_{n\le n_{\max}}\nu_n\bar\tau^{\bullet}(\delta,n)$\\
$\tau^{TATT}(\delta)$ & Aggregate total policy effect,
$\tau^{DATT}(\delta)+\tau^{SATT}(\delta)$\\
$\tau_{B,j}^\bullet(\alpha,n),\ \bar\tau_B^\bullet(\alpha,n),\ \tau_B^\bullet(\alpha)$ & Unit-level, size-conditional average, and aggregate
type-B policy effects defined in Remark~\ref{rem:typeB}\\
\midrule

$\mathcal M_{NP}$ & Nonparametric i.i.d.\ cluster model subject to the conditional-independence restriction \eqref{eq:assignCI}\\
$^*$ & Superscript denoting evaluation at the true observed-data distribution\\
$g_{j,a_{-j},j'}^\prime(x,n)$ & Policy derivative
$\partial\pi_\delta(a_{-j}\mid x,n)/\partial q_{j'}(x,n)$ in \eqref{eq:gprime}\\
$H_j^{DATT*}(x,n)$ & Policy average of
$\Delta m_j^*(a_{-j},x,n)$ under
$\pi_\delta^*(a_{-j}\mid x,n)$\\
$H_j^{SATT*}(x,n)$ & Policy average of
$m_{0,j}^*(a_{-j},x,n)-m_{0,j}^*(\mathbf0_{-j},x,n)$ under
$\pi_\delta^*(a_{-j}\mid x,n)$\\
$H_{B,j}^{DATT*}(x,n),\ H_{B,j}^{SATT*}(x,n)$ & Corresponding direct and spillover
regression contrasts averaged under the fixed type-B law $\pi_B$\\
$G_j^{DATT*}(O_i;\delta,n)$ & Policy-estimation term in \eqref{eq:Gdatt}\\
$G_j^{SATT*}(O_i;\delta,n)$ & Policy-estimation term in \eqref{eq:Gsatt}\\
$\varphi_j^{DATT*}(O_i;\delta,n)$ & Size-$n$ unit-level DATT EIF in \eqref{eq:phiDATT}\\
$\varphi_j^{SATT*}(O_i;\delta,n)$ & Size-$n$ unit-level SATT EIF in \eqref{eq:phiSATT}\\
$\bar\varphi^{\bullet*}(O_i;\delta)$ & Within-cluster EIF component,
$N_i^{-1}\sum_{j=1}^{N_i}\varphi_j^{\bullet*}(O_i;\delta,N_i)$\\
$\varphi^{\bullet*}(O_i;\delta)$ & Global EIF in \eqref{eq:eif}; its variance is the semiparametric efficiency bound\\
$\varphi_{B,j}^{\bullet*}(O_i;\alpha,n)$ & Type-B unit-level EIF in
\eqref{eq:phiBDATT}--\eqref{eq:phiBSATT}\\
$\varphi_B^{\bullet*}(O_i;\alpha)$ & Global type-B EIF in \eqref{eq:eifB}\\
\midrule

$\eta$ & Collection of nuisance functions and derived quantities
$(q_j,e_{a,j},m_{a,j},p_j,\tau_j^\bullet)$ over all admissible indices\\
$\Gamma^\bullet(O_i;\delta,\eta)$ & Uncentered EIF, or moment function, in \eqref{eq:momentfn}\\
$\Gamma_j^\bullet(O_i;\eta)$ & Unit-level uncentered EIF at fixed $(j,n)$,
with $(\delta,n)$ suppressed:
$\tau_j^\bullet(\delta,n;\eta)+\varphi_j^\bullet(O_i;\delta,n,\eta)$\\
$\eta_B$ & Type-B nuisance collection
$(q_j,e_{a,j},m_{a,j},p_j,\tau_{B,j}^\bullet)$\\
$\Gamma_B^\bullet,\ \Gamma_{B,j}^\bullet$ & Global and unit-level
uncentered type-B EIFs in \eqref{eq:momentfnB}\\
$K$ & Fixed number of cross-fitting folds\\
$\mathcal I_k$ & Set of cluster indices in fold $k$\\
$\mathcal I_k^c$ & Training complement of fold $k$\\
$\widehat\E_n^{(-k)}$ &
Empirical mean over all size-$n$ clusters in the full training complement $\mathcal I_k^c$\\
$\widehat\eta^{(-k)}$ & Nuisance estimates trained outside fold $k$\\
$\widehat q_j,\widehat e_{a,j},\widehat m_{a,j}$ & Estimated regression nuisances; a superscript $(-k)$ indicates fitting outside fold $k$\\
$\widehat\pi_\delta^{(-k)}$ & Estimated CIPS peer law, the functional \eqref{eq:pidelta} evaluated at $\widehat q^{(-k)}$\\
$\widehat p_j^{(-k)}(n)$ & Fold-specific estimate of $p_j(n)$ in \eqref{eq:derivednuis}\\
$\widehat{\tau}_j^{\bullet,(-k)}(\delta,n)$ & Fold-specific derived unit-level policy effect in \eqref{eq:derivednuis}\\
$\widehat\tau_{B,j}^{\bullet,(-k)}(\alpha,n)$ & Fold-specific derived
unit-level type-B effect in \eqref{eq:derivednuisB}\\
$\widehat{\tau}^{\bullet}(\delta)$ & Cross-fitted one-step estimator in \eqref{eq:tauhat}\\
$\widehat\tau_B^\bullet(\alpha)$ & Cross-fitted type-B one-step estimator
in \eqref{eq:tauhatB}\\
$\|f\|_{2,n}$ & Size-conditional $L_2$ norm,
$\{\E_n[f(X_i)^2]\}^{1/2}$\\
$\|Z\|_2$ & Unconditional $L_2$ norm,
$\{\E[Z^2]\}^{1/2}$; for a cross-fitted quantity, the fitted training-sample nuisances are held fixed\\
$\varepsilon_q$ & Maximum size-conditional $L_2$ error of $\widehat q_j^{(-k)}$, over folds and $(j,n)$\\
$\varepsilon_e$ & Maximum size-conditional $L_2$ error of $\widehat e_{a,j}^{(-k)}$, over folds and $(a,j,a_{-j},n)$\\
$\varepsilon_m$ & Maximum size-conditional $L_2$ error of $\widehat m_{a,j}^{(-k)}$, over folds and $(a,j,a_{-j},n)$\\
$c'$ & Lower-bound constant for estimated probability nuisances\\
$C$ & Generic finite constant whose value may change from line to line\\
$\sigma_\bullet^2$ & Asymptotic variance,
$\E[\{\varphi^{\bullet*}(O_i;\delta)\}^2]$\\
$\widehat V^{\bullet}(\delta)$ & Cross-fitted variance estimator in \eqref{eq:Vhat}\\
$\sigma_{B,\bullet}^2,\ \widehat V_B^{\bullet}(\alpha)$ & Type-B asymptotic
variance $\E[\{\varphi_B^{\bullet*}(O_i;\alpha)\}^2]$ and its cross-fitted
estimator in \eqref{eq:VhatB}\\
$\gamma\in(0,1)$ & Nominal error probability of the confidence interval\\
$z_{1-\gamma/2}$ & Standard-normal $(1-\gamma/2)$ quantile\\
\midrule

$p(\cdot)$ & Density or probability mass function\\
$t$ & One-dimensional regular parametric-submodel index used only in~\ref{app:eif}; distinct from the period index in the main text\\
$\partial_t f;\ \dot f$ & Derivative of a submodel functional at $t=0$\\
$o=(y_1,y_2,a_i,x,n)$ & Generic realization of the full cluster data $O_i$\\
$w=(\Delta y,a,a_{-j},x,n)$ & Generic realization of the reduced data $W_{ij}$ at fixed $(j,n)$\\
$p(o;t)$ & Observed-data density along the parametric submodel\\
$S(O_i)$ & Full submodel score,
$\partial_t\log p(O_i;t)|_{t=0}$\\
$S_n(O_i)$ & Conditional-law score given $N_i=n$\\
$S_N(n)$ & Cluster-size score,
$\partial_t\log\nu_n(t)|_{t=0}$\\
$s_{j,n}(W_{ij})$ & Reduced score,
$\E_n[S_n(O_i)\mid W_{ij}]$\\
$p_{W,j,n}(w;t)$ & Reduced density of $W_{ij}$ given $N_i=n$ along the submodel\\
$p_{a,j,n}(\Delta y\mid a_{-j},x;t)$ & Conditional density of $\Delta Y_{ij}$ along
the submodel, given $(A_{ij},A_{i(-j)},X_i,N_i)=(a,a_{-j},x,n)$\\
$p^*(a,a_{-j}\mid x,n)$ & True joint treatment probability:
$q_j^*(x,n)e_{1,j}^*(a_{-j}\mid x,n)$ if $a=1$, and
$\{1-q_j^*(x,n)\}e_{0,j}^*(a_{-j}\mid x,n)$ if $a=0$\\
$p_{X,n}(x;t)$ & Conditional density of $X_i$ given $N_i=n$ along the submodel\\
$s_{a,j,n}$ & Conditional-outcome score for treatment value $a$\\
$s_{X,n}(X_i)$ & Score of the covariate law given $N_i=n$\\
$\mathcal T_{j,n}$ & Size-$n$ reduced tangent space in \eqref{eq:tangent}\\
$\mathcal T_N$ & Cluster-size tangent block,
$\{L_N(N_i):\E[L_N]=0\}$\\
$L_Y,L_A,L_X,L_N$ & Outcome, assignment, covariate, and cluster-size tangent components\\
\midrule

$d_j^{DATT}(a_{-j},x)$ & Direct regression contrast used in the EIF proof,
$\Delta m_j^*(a_{-j},x,n)$\\
$d_j^{SATT}(a_{-j},x)$ & Spillover regression contrast used in the EIF proof,
$m_{0,j}^*(a_{-j},x,n)-m_{0,j}^*(\mathbf0_{-j},x,n)$\\
$\mathcal N_j^{DATT},\mathcal N_j^{SATT}$ & DATT and SATT numerator functionals at fixed $(j,n)$\\
$\mathcal D_j(t),\ \mathcal D_j$ & Unit-level denominator along a submodel
and at the truth, with $\mathcal D_j=p_j^*(n)$\\
$R_{Y,j}^{DATT},R_{Y,j}^{SATT}$ & Outcome-score representers in
\eqref{eq:RYD} and \eqref{eq:RYS}\\
$\mathcal N_{B,j}^\bullet(t)$ & Type-B numerator functional along the
parametric submodel in~\ref{app:typeB}\\
$R_{B,Y,j}^\bullet$ & Type-B outcome-score representer obtained from
$R_{Y,j}^{DATT}$ or $R_{Y,j}^{SATT}$ by replacing $\pi_\delta^*$ with $\pi_B$\\
\midrule

$\Gamma(\cdot;\eta),\ \tau^*,\ \varphi^*$ & Abbreviations for
$\Gamma^\bullet(\cdot;\delta,\eta)$,
$\tau^\bullet(\delta)$, and
$\varphi^{\bullet*}(\cdot;\delta)$, respectively\\
$\E^{(k)}$ & Expectation over a new cluster with
$\widehat\eta^{(-k)}$ held fixed\\
$d_j^\bullet(a_{-j},x,n;m)$ & Generic regression contrast constructed from $m$:
$m_{1,j}-m_{0,j}$ for DATT and
$m_{0,j}(a_{-j})-m_{0,j}(\mathbf 0_{-j})$ for SATT\\
$H_j^\bullet(x;q,m)$ & Generic policy average of
$d_j^\bullet(a_{-j},x,n;m)$ under
$\pi_\delta(\cdot\mid x,n;q)$, used in Lemma~\ref{lem:cc_derived}\\
$c$ & Generic peer configuration in $\mathcal A(n-1)$ within the second-order calculation\\
$\widehat\pi_c(x),\ \pi_c^*(x)$ & Estimated and true policy probabilities
$\pi_\delta(c\mid x,n;\widehat q)$ and
$\pi_\delta(c\mid x,n;q^*)$\\
$\widehat d_c^\bullet,d_c^{\bullet*}$ & Estimated and true DATT or SATT regression contrasts\\
$R_{2,c}(x)$ & Second-order Taylor remainder for the policy map in \eqref{eq:cc_taylor}\\
$\mathcal R_Y^\bullet$ & Outcome-regression contribution to the conditional bias\\
$\mathcal R_\pi^\bullet$ & Policy-estimation contribution to the conditional bias\\
$\mathcal R_{B,Y}^\bullet$ & Type-B outcome-regression contribution to the
conditional bias in \eqref{eq:cc_exact_biasB}\\
$\Delta\Gamma_k(O)$ & Fold-$k$ moment error,
$\Gamma(O;\widehat\eta^{(-k)})-\Gamma(O;\eta^*)$\\
$R_k$ & Fold-$k$ conditional bias,
$\E^{(k)}[\Gamma(O;\widehat\eta^{(-k)})]-\tau^*$\\
$T_1,T_2,T_3$ & Leading true-EIF average, centered cross-fitting empirical-process remainder, and conditional-bias remainder, respectively, in \eqref{eq:cc_decomp}\\
\end{longtable}
\endgroup

\clearpage
\section{Proofs}
\label{app:proofs}

Throughout this appendix, $C$ denotes a generic finite positive constant
that may change from line to line.

\subsection{Proof of Proposition~\ref{prop:id} (Identification)}
\label{app:id}

Fix $N_i=n$ throughout.

\medskip
\noindent\textbf{DATT.} For each $a_{-j}$ and each $x$,
\begin{align*}
\theta_{j}^{DATT}(a_{-j},x,n)
&=\E[Y_{ij2}(1,a_{-j})-Y_{ij2}(0,a_{-j})\mid A_{ij}=1,A_{i(-j)}=a_{-j},X_i=x,N_i=n]\\
&=\E[Y_{ij2}(1,a_{-j})-Y_{ij1}(0,\mathbf 0_{-j})\mid A_{ij}=1,A_{i(-j)}=a_{-j},X_i=x,N_i=n]\\
&\quad-\E[Y_{ij2}(0,a_{-j})-Y_{ij1}(0,\mathbf 0_{-j})\mid A_{ij}=0,A_{i(-j)}=a_{-j},X_i=x,N_i=n]\\
&=m_{1,j}(a_{-j},x,n)-m_{0,j}(a_{-j},x,n),
\end{align*}
where the second equality adds and subtracts $\E[Y_{ij1}(0,\mathbf 0_{-j})\mid A_{ij}=1,A_{i(-j)}=a_{-j},X_i=x,N_i=n]$ and applies Assumption~\ref{asm:1}(iii) to switch the conditioning event from $A_{ij}=1$ to $A_{ij}=0$ in the control term, and the third equality follows from Assumption~\ref{asm:1}(i)--(ii) together with $\Delta Y_{ij}=Y_{ij2}-Y_{ij1}$. This is \eqref{eq:idDATT}.

\medskip
\noindent\textbf{SATT.} For each $a_{-j}$ and each $x$,
\begin{align*}
\theta_{j}^{SATT}(a_{-j},x,n)
&=\E[Y_{ij2}(0,a_{-j})-Y_{ij2}(0,\mathbf 0_{-j})\mid A_{ij}=1,A_{i(-j)}=a_{-j},X_i=x,N_i=n]\\
&=\E[Y_{ij2}(0,a_{-j})-Y_{ij1}(0,a_{-j})\mid A_{ij}=1,A_{i(-j)}=a_{-j},X_i=x,N_i=n]\\
&\quad-\E[Y_{ij2}(0,\mathbf 0_{-j})-Y_{ij1}(0,\mathbf 0_{-j})\mid A_{ij}=1,A_{i(-j)}=a_{-j},X_i=x,N_i=n]\\
&=\E[Y_{ij2}(0,a_{-j})-Y_{ij1}(0,a_{-j})\mid A_{ij}=0,A_{i(-j)}=a_{-j},X_i=x,N_i=n]\\
&\quad-\E[Y_{ij2}(0,\mathbf 0_{-j})-Y_{ij1}(0,\mathbf 0_{-j})\mid A_{ij}=0,A_{i(-j)}=\mathbf 0_{-j},X_i=x,N_i=n]\\
&=m_{0,j}(a_{-j},x,n)-m_{0,j}(\mathbf 0_{-j},x,n),
\end{align*}
where the second equality adds and subtracts $\E[Y_{ij1}(0,a_{-j})\mid\cdot]$ and $\E[Y_{ij1}(0,\mathbf 0_{-j})\mid\cdot]$ and uses no anticipation (Assumption~\ref{asm:1}(ii)), the third applies (CPT-S1) to the $Y(0,a_{-j})$ trend and (CPT-S2) to the $Y(0,\mathbf 0_{-j})$ trend, and the last uses Assumption~\ref{asm:1}(i)--(ii). This is \eqref{eq:idSATT}.

\medskip
\noindent\textbf{Aggregation.} Under overlap (Assumption~\ref{asm:1}(iv)) the propensity $q_{j}(x,n)$ is identified from the observed assignment mechanism, hence so are the shifted propensity $q_{j,\delta}$ of \eqref{eq:qdelta} and the peer law $\pi_\delta(a_{-j}\mid x,n)$ of \eqref{eq:pidelta}. Substituting \eqref{eq:idDATT}--\eqref{eq:idSATT} into \eqref{eq:tauj} and rewriting the treated conditional expectation with Bayes' rule \eqref{eq:bayes} gives \eqref{eq:idtauDATT}--\eqref{eq:idtauSATT}; all objects on their right-hand sides are functionals of the observed-data distribution. The size law $\nu_n$ is identified from the observed $N_i$. Substituting into \eqref{eq:tau} identifies $\tau^{DATT}(\delta)$ and $\tau^{SATT}(\delta)$. \hfill$\square$

\subsection{Proof of Theorem~\ref{thm:eif} (Global efficiency)}
\label{app:eif}
The proof follows the strategy in Sections A.2 and B.2 of the supplementary
material to \citet{lee2025efficient}, adapted to the covariate-conditional
treated estimands considered here.
The proof has two layers.  Steps 1 and 2 derive the unit-level
efficient influence functions for
$\tau_j^{DATT}(\delta,n)$ and $\tau_j^{SATT}(\delta,n)$, respectively, whereas
Step 3 aggregates these unit-level components over units and the random
cluster-size distribution to obtain the efficient influence function of
$\tau^\bullet(\delta)$.  At each layer, the argument verifies the three
defining properties of the canonical gradient: (i) for every regular
parametric submodel, the pathwise derivative equals the inner product of the
candidate influence function with the corresponding score,
$\dot\Psi=\E[\varphi S]$ (with the size-conditional version used in Steps 1
and 2); (ii) the candidate influence function has mean zero; and (iii) it
belongs to the relevant tangent space.

We prove the result along an arbitrary one-dimensional regular parametric
submodel of $\mathcal{M}_{NP}$; the argument applies componentwise to a
finite-dimensional submodel.  Let
\[
p(o;t)=p(y_1,y_2\mid a,x,n;t)
\prod_{k=1}^{n}q_k(x,n;t)^{a_k}\{1-q_k(x,n;t)\}^{1-a_k}
p(x\mid n;t)\,\nu_n(t)
\]
be a smooth submodel through the truth at $t=0$, and write
$S(O_i)=\partial_t\log p(O_i;t)|_{t=0}$ for its score.  Conditional on
$N_i=n$, let
\[
S_n(O_i)=\partial_t\log p(O_i\mid N_i=n;t)|_{t=0},
\qquad
S_N(n)=\partial_t\log\nu_n(t)|_{t=0}.
\]
Thus $S(O_i)=S_N(N_i)+S_{N_i}(O_i)$, with
$\E[S_N(N_i)]=0$ and $\E_n[S_n(O_i)]=0$. Fix $n\le n_{\max}$ and $j\le n$.  The score induced by the submodel on
$W_{ij}=(\Delta Y_{ij},A_{ij},A_{i(-j)},X_i,n)$ is
\[
s_{j,n}(W_{ij})=\E_n[S_n(O_i)\mid W_{ij}].
\]
Let $p_{a,j,n}(\Delta y\mid a_{-j},x;t)$ denote the conditional density of
$\Delta Y_{ij}$ under the conditioning event
\[
(A_{ij},A_{i(-j)},X_i,N_i)=(a,a_{-j},x,n),
\]
and set $a_i=(a,a_{-j})$. The reduced density then factorizes as
\[
p_{W,j,n}(w;t)
=p_{X,n}(x;t)p_{a,j,n}(\Delta y\mid a_{-j},x;t)
\prod_{k=1}^{n}q_k(x,n;t)^{a_k}\{1-q_k(x,n;t)\}^{1-a_k}.
\]
Consequently its score has the decomposition
\begin{equation}\label{eq:sn}
\begin{aligned}
s_{j,n}(W_{ij})
&=\left.\partial_t\log p_{W,j,n}(W_{ij};t)\right|_{t=0}\\
&=A_{ij}s_{1,j,n}+(1-A_{ij})s_{0,j,n}
+\sum_{k=1}^{n}\frac{A_{ik}-q_k^*(X_i,n)}
{q_k^*(X_i,n)\{1-q_k^*(X_i,n)\}}\,\dot q_k(X_i,n)
+s_{X,n}(X_i),
\end{aligned}
\end{equation}
where a dot denotes $\partial_t|_{t=0}$, the arguments of the component
scores are suppressed in \eqref{eq:sn}, and
\[
\begin{aligned}
s_{a,j,n}(\Delta y,a_{-j},x)
&=\left.\partial_t\log
p_{a,j,n}(\Delta y\mid a_{-j},x;t)\right|_{t=0},
\qquad a\in\{0,1\},\\
s_{X,n}(x)
&=\left.\partial_t\log p_{X,n}(x;t)\right|_{t=0}.
\end{aligned}
\]
The component scores satisfy
\begin{equation}\label{eq:score_conditions}
\begin{gathered}
\E_n[s_{a,j,n}\mid A_{ij}=a,A_{i(-j)},X_i]=0, \quad
\E_n[s_{X,n}(X_i)]=0, \quad a\in\{0,1\}.
\end{gathered}
\end{equation}

Under $\mathcal{M}_{NP}$, the size-$n$ reduced tangent space is
\begin{equation}\label{eq:tangent}
\begin{aligned}
\mathcal T_{j,n}
=\biggl\{L_Y+\sum_{k=1}^{n}h_k(X_i)
\{A_{ik}-q_k^*(X_i,n)\}+L_X:\;&
\E_n[L_Y\mid A_{ij},A_{i(-j)},X_i]=0, \E_n[L_X]=0\biggr\},
\end{aligned}
\end{equation}
where $L_Y,L_X$, and every $h_k$ are square-integrable functions of their
displayed arguments.  The middle block is precisely the tangent space of the
conditionally independent Bernoulli assignment law
\eqref{eq:assignCI}.  The full tangent space additionally contains the
mutually orthogonal size block
$\mathcal T_N=\{L_N(N_i):\E[L_N]=0\}$.

We will repeatedly use two score-representer identities.  First, for every
$k\le n$ and every square-integrable $h$,
\begin{equation}\label{eq:representer}
\begin{aligned}
\int h(x)\dot q_k(x,n)p_{X,n}^*(x)\,dx
&=\E_n\!\left[h(X_i)\{A_{ik}-q_k^*(X_i,n)\}S_n(O_i)\right]\\
&=\E_n\!\left[h(X_i)\{A_{ik}-q_k^*(X_i,n)\}s_{j,n}(W_{ij})\right].
\end{aligned}
\end{equation}
Indeed, differentiating
$q_k(x,n;t)=\E_t[A_{ik}\mid X_i=x,N_i=n]$ by the quotient rule gives
\[
\begin{aligned}
\dot q_k(x,n)
&=\E_n[A_{ik}S_n(O_i)\mid X_i=x]
-q_k^*(x,n)\E_n[S_n(O_i)\mid X_i=x]\\
&=\E_n[\{A_{ik}-q_k^*(x,n)\}S_n(O_i)\mid X_i=x].
\end{aligned}
\]
Multiplying this identity by $h(x)$ and integrating with respect to
$p_{X,n}^*(x)$ gives the first equality in \eqref{eq:representer}.
The second equality follows from the defining projection identity for the
reduced score, because
$h(X_i)\{A_{ik}-q_k^*(X_i,n)\}$ is measurable with respect to $W_{ij}$:
$X_i$ and the full treatment vector
$A_i=(A_{ij},A_{i(-j)})$ are both contained in $W_{ij}$.  Second, for
$a\in\{0,1\}$, $a_{-j}\in\mathcal A(n-1)$, and square-integrable
$c(\cdot)$,
\begin{equation}\label{eq:mrep}
\begin{aligned}
&\int c(x)\dot m_{a,j}(a_{-j},x,n)p_{X,n}^*(x)\,dx\\
&=\E_n\!\left[
c(X_i)\frac{\one\{A_{ij}=a,A_{i(-j)}=a_{-j}\}}
{p^*(a,a_{-j}\mid X_i,n)}
\{\Delta Y_{ij}-m_{a,j}^*(a_{-j},X_i,n)\}S_n(O_i)\right]\\
&=\E_n\!\left[
c(X_i)\frac{\one\{A_{ij}=a,A_{i(-j)}=a_{-j}\}}
{p^*(a,a_{-j}\mid X_i,n)}
\{\Delta Y_{ij}-m_{a,j}^*(a_{-j},X_i,n)\}s_{j,n}(W_{ij})\right],
\end{aligned}
\end{equation}
where
$p^*(1,a_{-j}\mid x,n)=q_j^*(x,n)e_{1,j}^*(a_{-j}\mid x,n)$ and
$p^*(0,a_{-j}\mid x,n)=\{1-q_j^*(x,n)\}e_{0,j}^*(a_{-j}\mid x,n)$.
To verify \eqref{eq:mrep}, differentiating the relevant conditional outcome
mean by the quotient rule gives
\[
\begin{aligned}
\dot m_{a,j}(a_{-j},x,n)
=\E_n\!\left[
\{\Delta Y_{ij}-m_{a,j}^*(a_{-j},x,n)\}S_n(O_i)
\;\middle|\;
A_{ij}=a,A_{i(-j)}=a_{-j},X_i=x
\right].
\end{aligned}
\]
Multiplying by $c(x)p_{X,n}^*(x)$, integrating over $x$, and rewriting the
conditional expectation using the indicator of
$\{A_{ij}=a,A_{i(-j)}=a_{-j}\}$ and its conditional probability
$p^*(a,a_{-j}\mid X_i,n)$ gives the first equality in \eqref{eq:mrep}.
The random variable multiplying $S_n(O_i)$ there is measurable with respect
to $W_{ij}$, so the defining projection identity
$s_{j,n}(W_{ij})=\E_n[S_n(O_i)\mid W_{ij}]$ gives the second equality.

\medskip
\noindent\textbf{Step 1: The EIF for unit-level
DATT.}
For concision, throughout Steps 1 and 2 all functions are evaluated at size
$n$.  Along the submodel, define
\[
\begin{aligned}
d_j^{DATT}(a_{-j},x;t)
&=m_{1,j}(a_{-j},x,n;t)-m_{0,j}(a_{-j},x,n;t),\\
H_j^{DATT}(x;t)
&=\sum_{a_{-j}}\pi_\delta(a_{-j}\mid x,n;t)
d_j^{DATT}(a_{-j},x;t),\\
\mathcal N_j^{DATT}(t)
&=\int q_j(x,n;t)H_j^{DATT}(x;t)p_{X,n}(x;t)\,dx.
\end{aligned}
\]
At the truth, suppress the argument $t=0$ and write
\[
\begin{gathered}
d_j^{DATT}(a_{-j},x)=\Delta m_j^*(a_{-j},x,n),\qquad
H_j^{DATT}(x)=H_j^{DATT*}(x,n),\\
\mathcal N_j^{DATT}
=\E_n[q_j^*(X_i,n)H_j^{DATT}(X_i)],\qquad
\mathcal D_j=p_j^*(n)=\E_n[q_j^*(X_i,n)],\qquad
\tau_j^{DATT*}=\mathcal N_j^{DATT}/\mathcal D_j.
\end{gathered}
\]
The relevant derivatives satisfy
\[
\dot p_{X,n}=p_{X,n}^*s_{X,n},\qquad
\dot d_j^{DATT}=\dot m_{1,j}-\dot m_{0,j},\qquad
\dot\pi_\delta(a_{-j}\mid x,n)
=\sum_{j'\ne j}g_{j,a_{-j},j'}^{\prime*}(x,n)\dot q_{j'}(x,n).
\]
Applying the product rule first to $\mathcal N_j^{DATT}(t)$ and then to
$H_j^{DATT}(x;t)$ gives
\begin{align*}
\dot{\mathcal N}_j^{DATT}
={}&\int \dot q_j(x,n)H_j^{DATT}(x)p_{X,n}^*(x)\,dx\\
&+\int q_j^*(x,n)\dot H_j^{DATT}(x)p_{X,n}^*(x)\,dx
+\int q_j^*(x,n)H_j^{DATT}(x)\dot p_{X,n}(x)\,dx,\\
\dot H_j^{DATT}(x)
={}&\sum_{a_{-j}}\left[
\dot\pi_\delta(a_{-j}\mid x,n)d_j^{DATT}(a_{-j},x)
+\pi_\delta^*(a_{-j}\mid x,n)\dot d_j^{DATT}(a_{-j},x)
\right].
\end{align*}
Hence, substituting the derivative identities above produces the following
four channels:
\begin{align}
\dot{\mathcal N}_j^{DATT}
={}&\;\underbrace{\int q_j^*H_j^{DATT}s_{X,n}p_{X,n}^*\,dx}_{\mathrm{(I)}\ X}
+\underbrace{\int H_j^{DATT}\dot q_jp_{X,n}^*\,dx}_{\mathrm{(II)}\ q_j}\nonumber\\
&+\underbrace{\int q_j^*\sum_{a_{-j}}d_j^{DATT}(a_{-j},x)
\sum_{j'\ne j}g_{j,a_{-j},j'}^{\prime*}(x,n)\dot q_{j'}(x,n)
p_{X,n}^*(x)\,dx}_{\mathrm{(III)}\ \text{policy}}\nonumber\\
&+\underbrace{\int q_j^*\sum_{a_{-j}}\pi_\delta^*(a_{-j}\mid x,n)
\{\dot m_{1,j}(a_{-j},x,n)-\dot m_{0,j}(a_{-j},x,n)\}
p_{X,n}^*(x)\,dx}_{\mathrm{(IV)}\ \text{outcome}}.
\label{eq:Dchannels}
\end{align}
Here and below unlabeled integrands are evaluated at $(x,n)$.  Channels (I)
and (II) are represented by
\[
q_j^*(X_i,n)H_j^{DATT}(X_i)-\mathcal N_j^{DATT}
\qquad\text{and}\qquad
\{A_{ij}-q_j^*(X_i,n)\}H_j^{DATT}(X_i),
\]
respectively, the second by \eqref{eq:representer}; the two sum to
$A_{ij}H_j^{DATT}(X_i)-\mathcal N_j^{DATT}$.  Applying
\eqref{eq:representer} separately to every $j'\ne j$ in channel (III)
gives precisely $G_j^{DATT*}(O_i;\delta,n)$ in \eqref{eq:Gdatt}.
For channel (IV), apply \eqref{eq:mrep} to every configuration and then sum
over $a_{-j}$.  The configuration indicators collapse the sums to the
observed $A_{i(-j)}$, giving
\begin{equation}\label{eq:RYD}
\begin{aligned}
R_{Y,j}^{DATT}(O_i)
=&\;A_{ij}\frac{\pi_\delta^*(A_{i(-j)}\mid X_i,n)}
{e_{1,j}^*(A_{i(-j)}\mid X_i,n)}
\{\Delta Y_{ij}-m_{1,j}^*(A_{i(-j)},X_i,n)\}\\
&-(1-A_{ij})\frac{q_j^*(X_i,n)}{1-q_j^*(X_i,n)}
\frac{\pi_\delta^*(A_{i(-j)}\mid X_i,n)}
{e_{0,j}^*(A_{i(-j)}\mid X_i,n)}
\{\Delta Y_{ij}-m_{0,j}^*(A_{i(-j)},X_i,n)\}.
\end{aligned}
\end{equation}
Combining the four channels therefore yields
\begin{equation}\label{eq:NDgrad}
\dot{\mathcal N}_j^{DATT}
=\E_n\!\left[
\{A_{ij}H_j^{DATT}(X_i)-\mathcal N_j^{DATT}
+R_{Y,j}^{DATT}(O_i)+G_j^{DATT*}(O_i;\delta,n)\}S_n(O_i)\right].
\end{equation}
The expression inside braces has mean zero: the outcome term has conditional
mean zero given $(A_i,X_i,N_i=n)$, the policy term has conditional mean zero
given $(X_i,N_i=n)$, and
$\E_n[A_{ij}H_j^{DATT}(X_i)]=\mathcal N_j^{DATT}$. Because $\mathcal D_j=\E_n[A_{ij}]$, its size-$n$ gradient is
$A_{ij}-p_j^*(n)$.  By the ratio rule and the identity 
$\mathcal N_j^{DATT}=p_j^*(n)\tau_j^{DATT*}$, the conjectured EIF of $\tau_{j}^{DATT*}$ is obtained as follows:
\[
\begin{aligned}
\frac{1}{p_j^*(n)}
&\bigl[A_{ij}H_j^{DATT}(X_i)-\mathcal N_j^{DATT}+R_{Y,j}^{DATT}+G_j^{DATT*}
-\tau_j^{DATT*}\{A_{ij}-p_j^*(n)\}\bigr]\\
&=\frac{1}{p_j^*(n)}
\bigl[A_{ij}\{H_j^{DATT}(X_i)-\tau_j^{DATT*}\}
+R_{Y,j}^{DATT}(O_i)+G_j^{DATT*}(O_i;\delta,n)\bigr],
\end{aligned}
\]
which is exactly $\varphi_j^{DATT*}$ in \eqref{eq:phiDATT}.
Indeed, \eqref{eq:NDgrad} represents the numerator derivative, while the
denominator derivative satisfies
$\dot{\mathcal D}_j
=\E_n[\{A_{ij}-p_j^*(n)\}S_n(O_i)]$.  Therefore, for every regular
submodel, the quotient rule gives
\[
\begin{aligned}
\left.\partial_t\tau_j^{DATT}(\delta,n;t)\right|_{t=0}
&=\frac{1}{p_j^*(n)}
\{\dot{\mathcal N}_j^{DATT}-\tau_j^{DATT*}\dot{\mathcal D}_j\}\\
&=\frac{1}{p_j^*(n)}\E_n\!\left[
\begin{aligned}
\bigl\{&
A_{ij}H_j^{DATT}(X_i)-\mathcal N_j^{DATT}+R_{Y,j}^{DATT}(O_i)\\
&+G_j^{DATT*}(O_i;\delta,n)
-\tau_j^{DATT*}\{A_{ij}-p_j^*(n)\}\bigr\}S_n(O_i)
\end{aligned}
\right]\\
&=\E_n[\varphi_j^{DATT*}(O_i;\delta,n)S_n(O_i)].
\end{aligned}
\]
Next we verify that $\varphi_j^{DATT*}$ lies in the tangent space and has
mean zero.
Multiplying $\varphi_j^{DATT*}$ by
$p_j^*(n)$ and decomposing the treated-centering term gives
\begin{equation}\label{eq:Dtangentdecomp}
\begin{gathered}
p_j^*(n)\varphi_j^{DATT*}=L_Y+L_A+L_X,\\
L_Y=R_{Y,j}^{DATT},\qquad
L_A=\{A_{ij}-q_j^*(X_i,n)\}
\{H_j^{DATT}(X_i)-\tau_j^{DATT*}\}+G_j^{DATT*},\\
L_X=q_j^*(X_i,n)\{H_j^{DATT}(X_i)-\tau_j^{DATT*}\}.
\end{gathered}
\end{equation}
The first block has conditional mean zero given $(A_i,X_i,N_i=n)$.
By the definition of $G_j^{DATT*}$, the second is a sum of
$X_i$-measurable multiples of $A_{ik}-q_k^*(X_i,n)$ and hence lies in the
restricted assignment tangent block of \eqref{eq:tangent}; it also has
conditional mean zero given $(X_i,N_i=n)$.  The last has size-conditional
mean zero because
$\E_n[q_j^*H_j^{DATT}]=p_j^*(n)\tau_j^{DATT*}$.
Thus \eqref{eq:Dtangentdecomp} lies in the tangent space
\eqref{eq:tangent}.  Since it represents the derivative for every regular
submodel and belongs to the tangent space, it is the canonical gradient. Finally, Assumption~\ref{asm:1}(iv), bounded $n$, and fixed
$\delta\in(0,\infty)$ bound every inverse probability and every derivative
$g_{j,a_{-j},j'}^{\prime*}$.  Assumption~\ref{asm:1}(v) then implies that all three
blocks in \eqref{eq:Dtangentdecomp} are square-integrable.  This also proves
$\E_n[\varphi_j^{DATT*}]=0$ and
$\E_n[(\varphi_j^{DATT*})^2]<\infty$.

\medskip
\noindent\textbf{Step 2: complete unit-level SATT calculation.}
Along the same submodel, define
\[
\begin{aligned}
d_j^{SATT}(a_{-j},x;t)
&=m_{0,j}(a_{-j},x,n;t)-m_{0,j}(\mathbf 0_{-j},x,n;t),\\
H_j^{SATT}(x;t)
&=\sum_{a_{-j}}\pi_\delta(a_{-j}\mid x,n;t)
d_j^{SATT}(a_{-j},x;t),\\
\mathcal N_j^{SATT}(t)
&=\int q_j(x,n;t)H_j^{SATT}(x;t)p_{X,n}(x;t)\,dx.
\end{aligned}
\]
At the truth, suppress the argument $t=0$ and write
\[
\begin{gathered}
d_j^{SATT}(a_{-j},x)
=m_{0,j}^*(a_{-j},x,n)-m_{0,j}^*(\mathbf 0_{-j},x,n),\qquad
H_j^{SATT}(x)=H_j^{SATT*}(x,n),\\
\mathcal N_j^{SATT}=\E_n[q_j^*(X_i,n)H_j^{SATT}(X_i)],
\qquad
\tau_j^{SATT*}=\mathcal N_j^{SATT}/p_j^*(n).
\end{gathered}
\]
The derivatives of the spillover contrast and the policy average are
\[
\begin{aligned}
\dot d_j^{SATT}(a_{-j},x)
&=\dot m_{0,j}(a_{-j},x,n)
-\dot m_{0,j}(\mathbf 0_{-j},x,n),\\
\dot H_j^{SATT}(x)
&=\sum_{a_{-j}}\left[
\dot\pi_\delta(a_{-j}\mid x,n)d_j^{SATT}(a_{-j},x)
+\pi_\delta^*(a_{-j}\mid x,n)\dot d_j^{SATT}(a_{-j},x)
\right].
\end{aligned}
\]
Applying the product rule to $\mathcal N_j^{SATT}(t)$ and substituting these
identities gives all four derivative channels explicitly:
\begin{equation}\label{eq:Schannels}
\begin{aligned}
\dot{\mathcal N}_j^{SATT}
={}&\;\underbrace{\int q_j^*H_j^{SATT}s_{X,n}p_{X,n}^*\,dx}_{\mathrm{(I)}\ X}
+\underbrace{\int H_j^{SATT}\dot q_jp_{X,n}^*\,dx}_{\mathrm{(II)}\ q_j}\\
&+\underbrace{\int q_j^*\sum_{a_{-j}}d_j^{SATT}(a_{-j},x)
\sum_{j'\ne j}g_{j,a_{-j},j'}^{\prime*}(x,n)\dot q_{j'}(x,n)
p_{X,n}^*(x)\,dx}_{\mathrm{(III)}\ \text{policy}}\\
&+\underbrace{\int q_j^*\sum_{a_{-j}}\pi_\delta^*(a_{-j}\mid x,n)
\{\dot m_{0,j}(a_{-j},x,n)
-\dot m_{0,j}(\mathbf 0_{-j},x,n)\}
p_{X,n}^*(x)\,dx}_{\mathrm{(IV)}\ \text{outcome}}.
\end{aligned}
\end{equation}
Channels (I) and (II) are represented by
$q_j^*(X_i,n)H_j^{SATT}(X_i)-\mathcal N_j^{SATT}$ and
$\{A_{ij}-q_j^*(X_i,n)\}H_j^{SATT}(X_i)$, respectively, and hence sum to
$A_{ij}H_j^{SATT}(X_i)-\mathcal N_j^{SATT}$.  Applying
\eqref{eq:representer} configuration by configuration to channel (III)
gives $G_j^{SATT*}(O_i;\delta,n)$.  Finally, using
$\sum_{a_{-j}}\pi_\delta^*(a_{-j}\mid x,n)=1$, channel (IV) can be written as
\begin{equation}\label{eq:Soutcomechannel}
\begin{aligned}
\dot{\mathcal N}_{j,\mathrm{out}}^{SATT}
=&\;\sum_{a_{-j}}\int q_j^*(x,n)\pi_\delta^*(a_{-j}\mid x,n)
\dot m_{0,j}(a_{-j},x,n)p_{X,n}^*(x)\,dx\\
&-\int q_j^*(x,n)\dot m_{0,j}(\mathbf 0_{-j},x,n)p_{X,n}^*(x)\,dx.
\end{aligned}
\end{equation}
Applying \eqref{eq:mrep} to the first line configuration by configuration
and to the second line at the zero configuration gives the expression
\begin{equation}\label{eq:RYS}
\begin{aligned}
R_{Y,j}^{SATT}(O_i)
=&(1-A_{ij})\frac{q_j^*(X_i,n)}{1-q_j^*(X_i,n)}
\frac{\pi_\delta^*(A_{i(-j)}\mid X_i,n)}
{e_{0,j}^*(A_{i(-j)}\mid X_i,n)}
\{\Delta Y_{ij}-m_{0,j}^*(A_{i(-j)},X_i,n)\}\\
&-(1-A_{ij})\one\{A_{i(-j)}=\mathbf 0_{-j}\}
\frac{q_j^*(X_i,n)}{1-q_j^*(X_i,n)}
\frac{1}{e_{0,j}^*(\mathbf 0_{-j}\mid X_i,n)}
\{\Delta Y_{ij}-m_{0,j}^*(\mathbf 0_{-j},X_i,n)\}.
\end{aligned}
\end{equation}
Thus the four derivative channels give
\[
\dot{\mathcal N}_j^{SATT}
=\E_n\!\left[
\{A_{ij}H_j^{SATT}(X_i)-\mathcal N_j^{SATT}
+R_{Y,j}^{SATT}(O_i)+G_j^{SATT*}(O_i;\delta,n)\}S_n(O_i)\right].
\]
Combining this numerator gradient with the same denominator gradient
$A_{ij}-p_j^*(n)$ yields
\[
\varphi_j^{SATT*}
=\frac{1}{p_j^*(n)}
\bigl[A_{ij}\{H_j^{SATT}(X_i)-\tau_j^{SATT*}\}
+R_{Y,j}^{SATT}(O_i)+G_j^{SATT*}(O_i;\delta,n)\bigr],
\]
which is exactly \eqref{eq:phiSATT}.  Its tangent-space decomposition is
\[
\begin{gathered}
p_j^*(n)\varphi_j^{SATT*}=L_Y+L_A+L_X,\\
L_Y=R_{Y,j}^{SATT},\qquad
L_A=\{A_{ij}-q_j^*(X_i,n)\}
\{H_j^{SATT}(X_i)-\tau_j^{SATT*}\}+G_j^{SATT*},\\
L_X=q_j^*(X_i,n)\{H_j^{SATT}(X_i)-\tau_j^{SATT*}\}.
\end{gathered}
\]
The three conditional mean-zero properties and the square-integrability
argument are exactly as verified after \eqref{eq:Dtangentdecomp}.  Hence
$\varphi_j^{SATT*}$ is the size-$n$ unit-level canonical gradient,
has conditional mean zero, and has finite second moment.

\medskip
\noindent\textbf{Step 3: aggregation over units and random cluster size.}
For either $\bullet\in\{DATT,SATT\}$, set
\[
\bar\tau^{\bullet*}(\delta,n)
=\frac1n\sum_{j=1}^n\tau_j^{\bullet*}(\delta,n),
\qquad
\bar\varphi^{\bullet*}(O_i;\delta)
=\frac1{N_i}\sum_{j=1}^{N_i}
\varphi_j^{\bullet*}(O_i;\delta,N_i).
\]
The target along the submodel is the mixture
$\tau^{\bullet}(\delta;t)
=\sum_n\nu_n(t)\bar\tau^{\bullet}(\delta,n;t)$, so
\begin{equation}\label{eq:globalderivative}
\dot{\tau}^{\bullet}
=\sum_n\dot\nu_n\,\bar\tau^{\bullet*}(\delta,n)
+\sum_n\nu_n\,\dot{\bar\tau}^{\bullet}(\delta,n).
\end{equation}
Let
\[
\varphi^{\bullet*}
=\bar\varphi^{\bullet*}
+\bar\tau^{\bullet*}(\delta,N_i)
-\tau^{\bullet*}(\delta),
\]
denote the conjectured EIF of $\tau^{\bullet}$,
which is precisely the expression given in \eqref{eq:eif}.

We first check the identity $\E[\varphi^{\bullet*} S(O_i)]=\dot{\tau}^{\bullet}$. 
For the size-score channel, conditional mean zero of every unit-level
gradient, $\dot\nu_n=\nu_nS_N(n)$, and $\sum_n\dot\nu_n=0$ imply
\[
\begin{aligned}
\E[\varphi^{\bullet*}S_N(N_i)]
&=\sum_n\nu_n
\{\bar\tau^{\bullet*}(\delta,n)-\tau^{\bullet*}(\delta)\}
S_N(n)\\
&=\sum_n\dot\nu_n\,\bar\tau^{\bullet*}(\delta,n).
\end{aligned}
\]
For the conditional-law channel, the size-specific centering term drops
because $\E_n[S_n]=0$.  Since each
$\varphi_j^{\bullet*}$ is measurable with respect to $W_{ij}$, the
projection identity defining $s_{j,n}$ and Steps 1 and 2 yield
\[
\begin{aligned}
\E[\varphi^{\bullet*}S_{N_i}(O_i)]
&=\sum_n\nu_n\frac1n\sum_{j=1}^n
\E_n[\varphi_j^{\bullet*}(O_i;\delta,n)S_n(O_i)]\\
&=\sum_n\nu_n\frac1n\sum_{j=1}^n
\dot{\tau}_j^{\bullet}(\delta,n)
=\sum_n\nu_n\dot{\bar\tau}^{\bullet}(\delta,n).
\end{aligned}
\]
Adding the last two displays matches both terms of
\eqref{eq:globalderivative} for every regular submodel.
Moreover,
\[
\E[\varphi^{\bullet*}]
=\sum_n\nu_n\E_n[\bar\varphi^{\bullet*}]
+\sum_n\nu_n
\{\bar\tau^{\bullet*}(\delta,n)-\tau^{\bullet*}(\delta)\}
=0.
\]
The averages of the unit-level $L_Y,L_A,L_X$ blocks remain, respectively,
in the full outcome, assignment, and covariate tangent blocks, while
$\bar\tau^{\bullet*}(\delta,N_i)-\tau^{\bullet*}(\delta)$
belongs to $\mathcal T_N$.  Hence the candidate belongs to the full tangent
space.  Bounded cluster size and the unit-level second-moment bounds imply
$\E[(\varphi^{\bullet*})^2]<\infty$.  It is therefore the canonical
gradient, and its variance
$\E[\{\varphi^{\bullet*}(O_i;\delta)\}^2]$ is the semiparametric
efficiency bound. \hfill$\square$

\subsection{Proof of Theorem~\ref{thm:asymp}}
\label{app:asymp}
The proof follows the cross-fitting and remainder arguments in Sections
A.3--A.5 of the supplementary material to \citet{lee2025efficient}, adapted
to the nuisance structure and treated normalization of the present
estimands. Fix $\bullet\in\{DATT,SATT\}$ and $\delta$, and abbreviate
\[
\Gamma(\cdot;\eta)
=\Gamma^\bullet(\cdot;\delta,\eta),\qquad
\tau^*=\tau^\bullet(\delta),\qquad
\varphi^*=\varphi^{\bullet*}(\cdot;\delta).
\]
For fold $k$, let $\E^{(k)}$ denote expectation over a new cluster with
$\widehat\eta^{(-k)}$ held fixed.  All bounds below are uniform over the
finitely many $(j,a_{-j},n)$ and may depend only on the bounds in
Assumptions~\ref{asm:1}--\ref{asm:2}, $n_{\max}$, and $\delta$.
By Theorem~\ref{thm:eif},
$\E[\varphi^*]=0$ and
$\sigma_\bullet^2=\E[(\varphi^*)^2]<\infty$. Throughout the proofs in this subsection, the fold superscript
$(-k)$ is suppressed on fitted nuisances whenever the fold plays no role; it
is retained in the statements of the lemmas and of the theorem.
The proof uses the following three lemmas.

\begin{lem}[Bound of the uncentered efficient function]\label{lem:cc_lip}
Under Assumptions~\ref{asm:1} and~\ref{asm:2},
\begin{equation}\label{eq:cc_lip}
\begin{aligned}
\|\Gamma(\cdot;\widehat\eta^{(-k)})
-\Gamma(\cdot;\eta^*)\|_2
\le C\biggl(\varepsilon_q+\varepsilon_e+\varepsilon_m+\max_{j,n}\{
|\widehat p_j^{(-k)}(n)-p_j^*(n)|
+|\widehat{\tau}_j^{\bullet,(-k)}(\delta,n)
-\tau_j^{\bullet*}(\delta,n)|\}\biggr).
\end{aligned}
\end{equation}
\end{lem}

\begin{proof}
Because $n\le n_{\max}$, it suffices to bound one of the finitely many
unit-level terms in \eqref{eq:momentfn}.  Write such a term as
$T(\eta)=c(\eta)R(\eta)$, where $c(\eta)$ is a product of uniformly bounded
factors drawn from $1/p_j$, $q_j/(1-q_j)$, $1/e_{a,j}$,
$\pi_\delta$, $g_{j,a_{-j},j'}^{\prime}$, $m_{a,j}$,
$\tau_j^\bullet$, and treatment indicators, and $R(\eta)$ is at most one
outcome residual $\Delta Y_{ij}-m_{a,j}(\eta)$.  Decompose
\[
T(\widehat\eta)-T(\eta^*)
=\{c(\widehat\eta)-c(\eta^*)\}R(\widehat\eta)
+c(\eta^*)\{R(\widehat\eta)-R(\eta^*)\}.
\]
If the term contains an outcome residual, then
$R(\widehat\eta)-R(\eta^*)=-(\widehat m_{a,j}-m_{a,j}^*)$.
Because $|c(\eta^*)|\le C$, the second term in the decomposition has
$L_2$ norm at most $C\varepsilon_m$.  If the term has no outcome residual,
take $R\equiv1$ and absorb all its nuisance dependence, including that
through $m_{a,j}$ and $\tau_j^\bullet$, into $c$; the second term then
vanishes.

It remains to control the coefficient difference.  On the interval
$[c\wedge c',1-c\wedge c']$, the maps $q\mapsto q_\delta$,
$q\mapsto\pi_\delta$, and
$q\mapsto g_{j,a_{-j},j'}^{\prime}$ are continuously differentiable with
bounded derivatives.  The maps $(q,e,p)\mapsto q/(1-q)$, $1/e$, and $1/p$
are likewise Lipschitz on the trimmed sets.  Boundedness of all factors and
the telescoping inequality
$\left|\prod_{\ell=1}^L u_\ell-\prod_{\ell=1}^L v_\ell\right|
\le C\sum_{\ell=1}^L|u_\ell-v_\ell|$,
therefore give the pointwise bound
\[
\begin{aligned}
|c(\widehat\eta)-c(\eta^*)|
\le C\bigl(|\widehat q-q^*|+|\widehat e-e^*|
+|\widehat m-m^*|+|\widehat p_j-p_j^*| +|\widehat\tau_j^\bullet-\tau_j^{\bullet*}|\bigr).
\end{aligned}
\]
By Assumptions~\ref{asm:1}(v) and~\ref{asm:2}(i),
$R(\widehat\eta)$ has a uniformly bounded conditional second moment.
Multiplying the last display by $R(\widehat\eta)$ and taking $L_2$ norms
thus yields $\varepsilon_q$, $\varepsilon_e$, $\varepsilon_m$, and the two
derived-nuisance errors on the right-hand side of \eqref{eq:cc_lip}.
Combining the two terms in the decomposition and summing over the finitely
many $(j,a_{-j},n)$ terms proves \eqref{eq:cc_lip}.
\end{proof}

\begin{lem}[Derived nuisances]\label{lem:cc_derived}
Under Assumptions~\ref{asm:1} and~\ref{asm:2}, the estimates in
\eqref{eq:derivednuis} satisfy, uniformly over folds,
\begin{equation}\label{eq:cc_derived}
\max_{j,n}|\widehat p_j^{(-k)}(n)-p_j^*(n)|=O_p(M^{-1/2}),
\qquad
\max_{j,n}
|\widehat{\tau}_j^{\bullet,(-k)}(\delta,n)-\tau_j^{\bullet*}(\delta,n)|
=o_p(1).
\end{equation}
\end{lem}

\begin{proof}
Write $\widehat\E_n^{(-k)}$ for the empirical mean over all size-$n$
clusters in the full training complement $\mathcal I_k^c$, and suppress
the fold superscript in the remaining notation.  For
$\widehat{\tau}_j^\bullet$, define
\[
H_j^\bullet(x;q,m)
=\sum_{a_{-j}}\pi_\delta(a_{-j}\mid x,n;q)
d_j^\bullet(a_{-j},x,n;m),
\]
where
$d_j^{DATT}=m_{1,j}-m_{0,j}$ and
$d_j^{SATT}=m_{0,j}(a_{-j})-m_{0,j}(\mathbf 0_{-j})$.
Abbreviate
\[
\widehat h_j^\bullet(x)
=\widehat q_j(x,n)H_j^\bullet(x;\widehat q,\widehat m),
\qquad
h_j^{\bullet*}(x)
=q_j^*(x,n)H_j^\bullet(x;q^*,m^*),
\]
and let
$\widehat\mu_j^\bullet=\widehat\E_n^{(-k)}[
\widehat h_j^\bullet(X_i)]$ and
$\mu_j^{\bullet*}=\E_n[h_j^{\bullet*}(X_i)]$.
Then the numerator error splits into a full-training-sample empirical term
and a nuisance bias:
\[
\widehat\mu_j^\bullet-\mu_j^{\bullet*}
=\underbrace{(\widehat\E_n^{(-k)}-\E_n)
[\widehat h_j^\bullet(X_i)]}_{o_p(1)}
+\underbrace{\E_n[
\widehat h_j^\bullet(X_i)-h_j^{\bullet*}(X_i)]}
_{\text{nuisance bias}}.
\]
The first term is $o_p(1)$ by Assumption~\ref{asm:2}(iii). 
For the nuisance bias, smoothness of $\pi_\delta$, boundedness, and the
triangle inequality give
\[
\|H_j^\bullet(\cdot;\widehat q,\widehat m)
-H_j^\bullet(\cdot;q^*,m^*)\|_{2,n}
\le C(\varepsilon_q+\varepsilon_m).
\]
Consequently, the Cauchy--Schwarz inequality and the boundedness of
$H_j^\bullet(\cdot;\widehat q,\widehat m)$ and $q_j^*$ yield
\[
\begin{aligned}
\left|\E_n[\widehat h_j^\bullet-h_j^{\bullet*}]\right|
&\le
\left|\E_n[(\widehat q_j-q_j^*)
H_j^\bullet(\cdot;\widehat q,\widehat m)]\right| +
\left|\E_n[q_j^*\{
H_j^\bullet(\cdot;\widehat q,\widehat m)
-H_j^\bullet(\cdot;q^*,m^*)\}]\right|\\
&\le C(\varepsilon_q+\varepsilon_m).
\end{aligned}
\]
Thus
$|\widehat\mu_j^\bullet-\mu_j^{\bullet*}|
\le C(\varepsilon_q+\varepsilon_m)+o_p(1)=o_p(1)$.
Finally, $\widehat p_j(n)=\widehat\E_n^{(-k)}[A_{ij}]$ is the empirical
frequency of an observed indicator over the size-$n$ clusters of
$\mathcal I_k^c$, with population value exactly $p_j^*(n)$ and no dependence on
any fitted nuisance, so $\widehat p_j(n)-p_j^*(n)=O_p(M^{-1/2})=o_p(1)$
by the central limit theorem.  The estimated denominator
$\widehat p_j(n)$ is bounded away from zero by
Assumption~\ref{asm:2}(i), whereas the population denominator $p_j^*(n)$
is bounded away from zero by Assumption~\ref{asm:1}(iv).  Applying the ratio identity
$\widehat{\tau}_j^\bullet-\tau_j^{\bullet*}
=(\widehat\mu_j^\bullet-\mu_j^{\bullet*})/\widehat p_j
+\mu_j^{\bullet*}(p_j^*-\widehat p_j)
/ \widehat p_jp_j^*$
proves \eqref{eq:cc_derived}.  Uniformity follows because the number of
folds and the number of $(j,n)$ pairs are fixed and finite.
\end{proof}

\begin{lem}[Second-order conditional bias]\label{lem:cc_remainder}
Let
$\Gamma_j^\bullet(O_i;\eta)
=\tau_j^\bullet(\delta,n;\eta)
+\varphi_j^\bullet(O_i;\delta,n,\eta)$.
Under Assumptions~\ref{asm:1} and~\ref{asm:2}, uniformly over $(j,n)$ and
folds,
\begin{equation}\label{eq:cc_remainder_bound}
\left|\E_n[\Gamma_j^\bullet(O_i;\widehat\eta^{(-k)})]
-\tau_j^{\bullet*}(\delta,n)\right|
\le C\{\varepsilon_q^2
+(\varepsilon_q+\varepsilon_e)\varepsilon_m\}
+o_p(M^{-1/2}).
\end{equation}
\end{lem}

\begin{proof}
Fix $(j,n)$ and suppress these indices where no confusion can arise.  Write
$\widehat q=\widehat q_j(\cdot,n)$, $q^*=q_j^*(\cdot,n)$,
$\widehat p=\widehat p_j(n)$, $p^*=p_j^*(n)$, and
$\widehat\tau_j=\widehat\tau_j^\bullet(\delta,n)$.  Direct substitution in
$\Gamma_j^\bullet=\widehat\tau_j+
\varphi_j^\bullet(\widehat\eta)$, followed by grouping the terms according
to their source, gives
\begin{equation}\label{eq:cc_exact_bias}
\E_n[\Gamma_j^\bullet(O_i;\widehat\eta)]
-\tau_j^{\bullet*}(\delta,n)
=\frac{1}{\widehat p}
\left[
(\widehat p-p^*)\{\widehat{\tau}_j
-\tau_j^{\bullet*}(\delta,n)\}
+\mathcal R_Y^\bullet+\mathcal R_\pi^\bullet
\right].
\end{equation}
Here $\mathcal R_Y^\bullet$ collects the terms containing
outcome regression errors, whereas $\mathcal R_\pi^\bullet$ collects the
terms caused by replacing the true CIPS policy probabilities with their
estimated versions.  Since $\widehat p$ is bounded away from zero,
\begin{equation}\label{eq:cc_bias_decomp}
\begin{aligned}
\left|\E_n[\Gamma_j^\bullet(O_i;\widehat\eta)]
-\tau_j^{\bullet*}(\delta,n)\right|
\le C\left[
|\widehat p-p^*|\,
|\widehat\tau_j-\tau_j^{\bullet*}(\delta,n)|
+|\mathcal R_Y^\bullet|+|\mathcal R_\pi^\bullet|
\right].
\end{aligned}
\end{equation}
We now bound the three terms on the right-hand side in order.

\medskip
\noindent\textit{1. Empirical-normalization term.}

By Lemma~\ref{lem:cc_derived},
\[
|\widehat p-p^*|\,
|\widehat{\tau}_j-\tau_j^{\bullet*}(\delta,n)|
=O_p(M^{-1/2})\,o_p(1)
=o_p(M^{-1/2}).
\]

\medskip
\noindent\textit{2. Outcome-regression remainder.}
For a peer configuration $c\in\mathcal A(n-1)$, write
$\widehat\pi_c(x)=\pi_\delta(c\mid x,n;\widehat q)$.
Iterated expectations applied to the outcome-augmentation terms give, for
DATT,
\begin{equation}\label{eq:cc_RYD}
\begin{aligned}
\mathcal R_Y^{DATT}
=&\;\E_n\sum_c\widehat\pi_c q^*
\left(1-\frac{e_{1,j}^*(c\mid X_i,n)}
{\widehat e_{1,j}(c\mid X_i,n)}\right)
\times
\{\widehat m_{1,j}(c,X_i,n)-m_{1,j}^*(c,X_i,n)\}\\
&+\E_n\sum_c\widehat\pi_c
\left\{
\frac{(1-q^*)e_{0,j}^*(c\mid X_i,n)\widehat q}
{(1-\widehat q)\widehat e_{0,j}(c\mid X_i,n)}
-q^*\right\}
\times
\{\widehat m_{0,j}(c,X_i,n)-m_{0,j}^*(c,X_i,n)\}.
\end{aligned}
\end{equation}
For SATT, the same calculation gives
\begin{equation}\label{eq:cc_RYS}
\begin{aligned}
\mathcal R_Y^{SATT}
=&\;\E_n\sum_c\widehat\pi_c
\left\{q^*-
\frac{(1-q^*)e_{0,j}^*(c\mid X_i,n)\widehat q}
{(1-\widehat q)\widehat e_{0,j}(c\mid X_i,n)}
\right\}\times
\{\widehat m_{0,j}(c,X_i,n)-m_{0,j}^*(c,X_i,n)\}\\
&+\E_n\left\{
\frac{(1-q^*)e_{0,j}^*(\mathbf 0_{-j}\mid X_i,n)\widehat q}
{(1-\widehat q)
\widehat e_{0,j}(\mathbf 0_{-j}\mid X_i,n)}
-q^*\right\}\times
\{\widehat m_{0,j}(\mathbf 0_{-j},X_i,n)
-m_{0,j}^*(\mathbf 0_{-j},X_i,n)\}.
\end{aligned}
\end{equation}
Thus every summand has the form
\[
\text{assignment-model error coefficient}
\times
\{\widehat m_{a,j}(c,X_i,n)-m_{a,j}^*(c,X_i,n)\}.
\]
By the same argument as in Lemma~\ref{lem:cc_lip}, we have
\begin{equation}\label{eq:cc_RY_bound}
|\mathcal R_Y^\bullet|
\le C(\varepsilon_q+\varepsilon_e)\varepsilon_m.
\end{equation}

\medskip
\noindent\textit{3. Policy-estimation remainder.}
Let $\widehat d_c^\bullet$ and $d_c^{\bullet*}$ denote the estimated and
true DATT or SATT regression contrasts, and let
$\pi_c^*(x)=\pi_\delta(c\mid x,n;q^*)$.  The policy contribution initially
has the exact expansion
\begin{equation}\label{eq:cc_Rpi_initial}
\begin{aligned}
\mathcal R_\pi^\bullet
=\E_n\left[
q^*\sum_c\{\widehat\pi_c-\pi_c^*\}d_c^{\bullet*}
\right] +\E_n\left[
\sum_c\widehat q\,\widehat d_c^\bullet
\sum_{j'\ne j}\widehat g_{j,c,j'}^{\prime}
\{q_{j'}^*-\widehat q_{j'}\}
\right].
\end{aligned}
\end{equation}
The first term is the error from replacing the true policy probability
$\pi_c^*$ by $\widehat\pi_c$.  The second term is
$\E_n[\widehat G_j^\bullet]$, because conditionally on $X_i$,
$\E_n[A_{ij'}-\widehat q_{j'}\mid X_i]
=q_{j'}^*-\widehat q_{j'}$.
To combine these two lines, take a second-order Taylor expansion of the
policy probability around $\widehat q$:
\begin{equation}\label{eq:cc_taylor}
\pi_c^*(x)-\widehat\pi_c(x)
=\sum_{j'\ne j}\widehat g_{j,c,j'}^{\prime}(x,n)
\{q_{j'}^*(x,n)-\widehat q_{j'}(x,n)\}+R_{2,c}(x),
\end{equation}
where
\begin{equation}\label{eq:cc_taylor_bound}
|R_{2,c}(x)|
\le C\sum_{j'\ne j}
\{q_{j'}^*(x,n)-\widehat q_{j'}(x,n)\}^2.
\end{equation}
Substituting \eqref{eq:cc_taylor} into the first term of
\eqref{eq:cc_Rpi_initial}, its first-order term combines with the second
term of \eqref{eq:cc_Rpi_initial}.  This leaves
\begin{equation}\label{eq:cc_Rpi}
\begin{aligned}
\mathcal R_\pi^\bullet
=&\;\E_n\sum_c
\{\widehat q\,\widehat d_c^\bullet-q^*d_c^{\bullet*}\}
\sum_{j'\ne j}\widehat g_{j,c,j'}^{\prime}
\{q_{j'}^*-\widehat q_{j'}\}\\
&-\E_n\left[q^*\sum_c d_c^{\bullet*}R_{2,c}\right].
\end{aligned}
\end{equation}
To bound the first line, note that
$\|\widehat q\,\widehat d_c^\bullet-q^*d_c^{\bullet*}\|_{2,n}
\le C(\varepsilon_q+\varepsilon_m)$.
The other factor has $L_2$ norm at most $C\varepsilon_q$.
The Cauchy--Schwarz inequality therefore gives
\[
\left|\E_n\sum_c
\{\widehat q\,\widehat d_c^\bullet-q^*d_c^{\bullet*}\}
\sum_{j'\ne j}\widehat g_{j,c,j'}^{\prime}
\{q_{j'}^*-\widehat q_{j'}\}\right|
\le C(\varepsilon_q^2+\varepsilon_q\varepsilon_m).
\]
For the second line of \eqref{eq:cc_Rpi},
\eqref{eq:cc_taylor_bound} gives
\[
\left|\E_n\left[q^*\sum_c d_c^{\bullet*}R_{2,c}\right]\right|
\le C\varepsilon_q^2.
\]
Hence
\begin{equation}\label{eq:cc_Rpi_bound}
|\mathcal R_\pi^\bullet|
\le C(\varepsilon_q^2+\varepsilon_q\varepsilon_m).
\end{equation}

Substituting the three bounds into \eqref{eq:cc_bias_decomp} gives
\[
\left|\E_n[\Gamma_j^\bullet(O_i;\widehat\eta)]
-\tau_j^{\bullet*}(\delta,n)\right|
\le C\{\varepsilon_q^2
+(\varepsilon_q+\varepsilon_e)\varepsilon_m\}
+o_p(M^{-1/2}),
\]
which proves \eqref{eq:cc_remainder_bound}.
\end{proof}

\noindent We now begin the formal proof of Theorem~\ref{thm:asymp}, which
proceeds in five steps.

\medskip
\noindent\textbf{Step 1: stochastic expansion.}
Let
\[
\Delta\Gamma_k(O)
=\Gamma(O;\widehat\eta^{(-k)})
-\Gamma(O;\eta^*),\qquad
R_k=\E^{(k)}[\Gamma(O;\widehat\eta^{(-k)})]-\tau^*.
\]
Because
$\Gamma(O;\eta^*)-\tau^*=\varphi^*(O)$,
\begin{equation}\label{eq:cc_decomp}
\begin{aligned}
\widehat{\tau}^\bullet(\delta)-\tau^*
=\underbrace{\frac1M\sum_{i=1}^M\varphi^*(O_i)}_{T_1}
+\underbrace{\frac1M\sum_{k=1}^K\sum_{i\in\mathcal I_k}
\left[\Delta\Gamma_k(O_i)
-\E^{(k)}\{\Delta\Gamma_k(O)\}\right]}_{T_2}
+\underbrace{\sum_{k=1}^K\frac{|\mathcal I_k|}{M}R_k}_{T_3}.
\end{aligned}
\end{equation}

\medskip
\noindent\textbf{Step 2: leading empirical-process term.}
The clusters are i.i.d., $\E[\varphi^*]=0$, and
$\E[(\varphi^*)^2]=\sigma_\bullet^2<\infty$.  Therefore the
Lindeberg--L\'evy central limit theorem gives
\[
\sqrt M\,T_1
=\frac1{\sqrt M}\sum_{i=1}^M\varphi^*(O_i)
\xrightarrow{d}\mathcal N(0,\sigma_\bullet^2).
\]

\medskip
\noindent\textbf{Step 3: cross-fitting empirical-process remainder.}
Conditionally on the training data for fold $k$, the summands indexed by
$i\in\mathcal I_k$ are i.i.d.\ and mean zero.  Hence
\[
\E\!\left[
\left.
\left\{\frac1{\sqrt M}\sum_{i\in\mathcal I_k}
\left[\Delta\Gamma_k(O_i)
-\E^{(k)}\Delta\Gamma_k(O)\right]\right\}^2
\right|\text{training data}\right]
\le \frac{|\mathcal I_k|}{M}
\|\Delta\Gamma_k\|_2^2.
\]

Lemmas~\ref{lem:cc_lip}--\ref{lem:cc_derived} and
Assumption~\ref{asm:2}(ii)--(iii) imply
$\|\Delta\Gamma_k\|_2=o_p(1)$.  Conditional Chebyshev's inequality
and fixed $K$ therefore give $\sqrt M\,T_2=o_p(1)$.

\medskip
\noindent\textbf{Step 4: conditional-bias remainder.}
By \eqref{eq:momentfn} and iterated expectations,
\[
R_k
=\sum_{n\le n_{\max}}\nu_n\,\frac1n\sum_{j=1}^n
\left\{
\E_n[\Gamma_j^\bullet(O_i;\widehat\eta^{(-k)})]
-\tau_j^{\bullet*}(\delta,n)
\right\}.
\]

Lemma~\ref{lem:cc_remainder}, the finite number of $(j,n)$ terms, and
Assumption~\ref{asm:2}(iii) yield
\[
|R_k|
\le C\{\varepsilon_q^2
+(\varepsilon_q+\varepsilon_e)\varepsilon_m\}
+o_p(M^{-1/2})
=o_p(M^{-1/2}).
\]
Since $K$ is fixed, $\sqrt M\,T_3=o_p(1)$.  Combining Steps 1--4 proves the
asymptotic linear representation and the normal limit.

\medskip
\noindent\textbf{Step 5: variance consistency.}
For $i\in\mathcal I_k$,
\[
\Gamma(O_i;\widehat\eta^{(-k)})
-\widehat{\tau}^{\bullet}(\delta)
=\varphi^*(O_i)+\Delta\Gamma_k(O_i)
-\{\widehat{\tau}^{\bullet}(\delta)-\tau^*\}.
\]
The law of large numbers gives
$M^{-1}\sum_i\varphi^*(O_i)^2\to_p\sigma_\bullet^2$.
The conditional second-moment calculation in Step 3 and
$\|\Delta\Gamma_k\|_2=o_p(1)$ give
\[
\frac1M\sum_{k=1}^K\sum_{i\in\mathcal I_k}
\Delta\Gamma_k(O_i)^2=o_p(1).
\]
Steps 1--4 imply
$\widehat{\tau}^{\bullet}(\delta)-\tau^*=o_p(1)$.
Expanding the square in \eqref{eq:Vhat} and applying Cauchy--Schwarz to all
cross terms proves
$\widehat V^{\bullet}(\delta)\to_p\sigma_\bullet^2$.
The confidence interval follows from Slutsky's theorem. \hfill$\square$

\subsection{Type-B policy: efficient influence functions, estimation, and
asymptotic properties}
\label{app:typeB}

\paragraph{Cross-fitted estimation and asymptotic properties}
The efficient influence function in Remark~\ref{rem:typeB} directly yields
a cross-fitted estimator of the type-B target.  Collect the relevant nuisances as
$\eta_B=(q_j,e_{a,j},m_{a,j},p_j,\tau_{B,j}^\bullet)$, let
$\varphi_{B,j}^\bullet(O_i;\alpha,n,\eta_B)$ denote
\eqref{eq:phiBDATT} or \eqref{eq:phiBSATT} with the starred quantities
replaced by their components of $\eta_B$, and define
\begin{equation}\label{eq:momentfnB}
\begin{aligned}
\Gamma_B^\bullet(O_i;\alpha,\eta_B)
&=\frac1{N_i}\sum_{j=1}^{N_i}
\Gamma_{B,j}^\bullet(O_i;\alpha,N_i,\eta_B),\\
\Gamma_{B,j}^\bullet(O_i;\alpha,n,\eta_B)
&=\tau_{B,j}^\bullet(\alpha,n;\eta_B)
+\varphi_{B,j}^\bullet(O_i;\alpha,n,\eta_B).
\end{aligned}
\end{equation}
At the truth,
$\E[\Gamma_B^\bullet(O_i;\alpha,\eta_B^*)]
=\tau_B^\bullet(\alpha)$.

For fold $k$, use the full training complement $\mathcal I_k^c$ to estimate
$(q_j,e_{a,j},m_{a,j})$ and to construct
$\widehat p_j^{(-k)}(n)$ as in \eqref{eq:derivednuis}.  Define

\begin{equation}\label{eq:derivednuisB}
\widehat\tau_{B,j}^{\bullet,(-k)}(\alpha,n)
=\frac{\widehat\E_n^{(-k)}\!\left[
\widehat q_j^{(-k)}(X_i,n)
\sum_{a_{-j}}\pi_B(a_{-j}\mid n;\alpha)
\widehat d_j^{\bullet,(-k)}(a_{-j},X_i,n)\right]}
{\widehat p_j^{(-k)}(n)},
\end{equation}
where
$\widehat d_j^{DATT,(-k)}
=\widehat m_{1,j}^{(-k)}-\widehat m_{0,j}^{(-k)}$ and
$\widehat d_j^{SATT,(-k)}
=\widehat m_{0,j}^{(-k)}(a_{-j})-\widehat m_{0,j}^{(-k)}(\mathbf 0_{-j})$.
The type-B one-step estimator is
\begin{equation}\label{eq:tauhatB}
\widehat\tau_B^\bullet(\alpha)
=\frac1M\sum_{k=1}^K\sum_{i\in\mathcal I_k}
\Gamma_B^\bullet(O_i;\alpha,\widehat\eta_B^{(-k)}).
\end{equation}
Thus \eqref{eq:tauhatB} is obtained from the CIPS estimator by replacing
the estimated covariate-dependent policy with the fixed weight $\pi_B$
and deleting $G_j^\bullet$.

The reduced rate condition follows from the conditional bias of
\eqref{eq:momentfnB}.  Fix $(j,n)$ and suppress the fold index.  Specializing
the exact bias identity \eqref{eq:cc_exact_bias} to the fixed policy gives
\begin{equation}\label{eq:cc_exact_biasB}
\begin{aligned}
\E_n[\Gamma_{B,j}^\bullet(O_i;\alpha,\widehat\eta_B)]
-\tau_{B,j}^{\bullet*}(\alpha,n)
=\frac1{\widehat p_j(n)}
\Bigl[&
\{\widehat p_j(n)-p_j^*(n)\}
\{\widehat\tau_{B,j}^\bullet(\alpha,n)
-\tau_{B,j}^{\bullet*}(\alpha,n)\} +\mathcal R_{B,Y}^\bullet\Bigr].
\end{aligned}
\end{equation}
Here $\mathcal R_{B,Y}^\bullet$ is the outcome remainder in
\eqref{eq:cc_RYD} or \eqref{eq:cc_RYS}, with
$\widehat\pi_c$ replaced by the fixed $\pi_B(c\mid n;\alpha)$.
By the same argument as in Lemma~\ref{lem:cc_lip}, we have
\[
|\mathcal R_{B,Y}^\bullet|
\le C(\varepsilon_q+\varepsilon_e)\varepsilon_m.
\]
There is no policy remainder $\mathcal R_\pi^\bullet$: since $\pi_B$ is
fixed, no Taylor expansion of the policy in $q$ is needed.  This removes
both the isolated $\varepsilon_q^2$ term and the
$\varepsilon_q\varepsilon_m$ contribution generated by estimation of the
CIPS policy.  The remaining $\varepsilon_q\varepsilon_m$ component inside
$(\varepsilon_q+\varepsilon_e)\varepsilon_m$ is the usual
assignment--outcome product.  Impose the same boundedness and
fixed-$K$ outer cross-fitting as in Theorem~\ref{thm:asymp}, together with the
analogous full-training-sample stability condition on the type-B numerator
integrands.  Since $\widehat p_j(n)$ is again an empirical frequency of
an observed indicator, the argument of Lemma~\ref{lem:cc_derived} gives
$\widehat p_j-p_j^*=O_p(M^{-1/2})$ and, with consistency of all regression
nuisances, $\widehat\tau_{B,j}^\bullet-\tau_{B,j}^{\bullet*}=o_p(1)$, hence
$\{\widehat p_j-p_j^*\}
\{\widehat\tau_{B,j}^\bullet-\tau_{B,j}^{\bullet*}\}
=O_p(M^{-1/2})o_p(1)=o_p(M^{-1/2})$.

Consequently the proof of Theorem~\ref{thm:asymp} applies with
\[
(\varepsilon_q+\varepsilon_e)\varepsilon_m=o_p(M^{-1/2}),
\]
and yields
\[
\sqrt M\{\widehat\tau_B^\bullet(\alpha)-\tau_B^\bullet(\alpha)\}
=\frac1{\sqrt M}\sum_{i=1}^M
\varphi_B^{\bullet*}(O_i;\alpha)+o_p(1)
\xrightarrow{d}\mathcal N(0,\sigma_{B,\bullet}^2),
\qquad
\sigma_{B,\bullet}^2=\E[\{\varphi_B^{\bullet*}(O_i;\alpha)\}^2],
\]
assuming $\sigma_{B,\bullet}^2>0$.
When $e_{a,j}$ is constructed from the product of the individual
propensities under \eqref{eq:assignCI}, bounded cluster size implies
$\varepsilon_e\le C\varepsilon_q$, so the rate condition simplifies to
$\varepsilon_q\varepsilon_m=o_p(M^{-1/2})$.

Inference proceeds exactly as in Theorem~\ref{thm:asymp} with
$\Gamma_B^\bullet$ in place of $\Gamma^\bullet$: setting
\begin{equation}\label{eq:VhatB}
\widehat V_B^{\bullet}(\alpha)
=\frac1M\sum_{k=1}^K\sum_{i\in\mathcal I_k}
\left\{\Gamma_B^\bullet(O_i;\alpha,\widehat\eta_B^{(-k)})
-\widehat\tau_B^{\bullet}(\alpha)\right\}^2,
\end{equation}
the argument of Step~5 of~\ref{app:asymp} applies verbatim and gives
$\widehat V_B^{\bullet}(\alpha)\to_p\sigma_{B,\bullet}^2$, so that
\[
\widehat\tau_B^{\bullet}(\alpha)
\ \pm\ z_{1-\gamma/2}
\sqrt{\widehat V_B^{\bullet}(\alpha)/M}
\]
is an asymptotically valid $(1-\gamma)$ confidence interval.  This is the
analytic standard error reported for the type-B estimators in
Section~\ref{sec:simulation}.

Equation~\eqref{eq:cc_exact_biasB} also displays double robustness.  If the
outcome regressions are correct, their errors vanish and hence
$\mathcal R_{B,Y}^\bullet=0$ regardless of the probability limit of the
assignment nuisance estimates.  If $(q_j,e_{a,j})$ are correct, the
coefficients multiplying the outcome-regression errors vanish, so
$\mathcal R_{B,Y}^\bullet=0$ regardless of the probability limit of the
outcome regressions.  Since $\widehat p_j$ is an empirical treated
proportion and the remaining quantities are bounded, the first product in
\eqref{eq:cc_exact_biasB} is $o_p(1)$ in either case.  Therefore
\eqref{eq:tauhatB} is consistent under either nuisance model.  This
union-model statement concerns consistency; root-$M$ inference when one
nuisance family is misspecified requires the corresponding additional
rate control.  The CIPS estimator itself is not doubly robust in this model-misspecification sense,
because an inconsistent $q_j$ changes the policy defining its estimand.

\paragraph{Proof of the type-B efficient influence function}
Fix $(j,n)$ and consider the same regular parametric submodel as in the
proof of Theorem~\ref{thm:eif}.  For
$\bullet\in\{DATT,SATT\}$, let
\[
\begin{aligned}
H_{B,j}^{DATT}(x;t)
&=\sum_{a_{-j}}\pi_B(a_{-j}\mid n;\alpha)
\{m_{1,j}(a_{-j},x,n;t)-m_{0,j}(a_{-j},x,n;t)\},\\
H_{B,j}^{SATT}(x;t)
&=\sum_{a_{-j}}\pi_B(a_{-j}\mid n;\alpha)
\{m_{0,j}(a_{-j},x,n;t)
-m_{0,j}(\mathbf 0_{-j},x,n;t)\},\\
\mathcal N_{B,j}^{\bullet}(t)
&=\int q_j(x,n;t)H_{B,j}^{\bullet}(x;t)
p_{X,n}(x;t)\,dx,\qquad
\mathcal D_j(t)=\int q_j(x,n;t)p_{X,n}(x;t)\,dx.
\end{aligned}
\]
At $t=0$,
$H_{B,j}^{\bullet}(x;0)=H_{B,j}^{\bullet*}(x,n)$,
$\mathcal D_j(0)=p_j^*(n)$, and
\[
\tau_{B,j}^{\bullet*}(\alpha,n)
=\frac{\mathcal N_{B,j}^{\bullet}(0)}{\mathcal D_j(0)}.
\]
The key difference from the CIPS calculation is that $\pi_B$ is fixed
along every parametric submodel.  Therefore
\[
\dot\pi_B(a_{-j}\mid n;\alpha)=0
\]
and
\[
\dot H_{B,j}^{\bullet}(x)
=\sum_{a_{-j}}\pi_B(a_{-j}\mid n;\alpha)
\dot d_j^\bullet(a_{-j},x),
\]
where
$\dot d_j^\bullet=\dot m_{1,j}-\dot m_{0,j}$ for
$\bullet=DATT$, and
$\dot d_j^\bullet=\dot m_{0,j}(a_{-j})
-\dot m_{0,j}(\mathbf 0_{-j})$ for $\bullet=SATT$.
Thus the product rule gives only the covariate-law,
focal-propensity, and outcome channels:
\begin{equation}\label{eq:NBchannels}
\begin{aligned}
\dot{\mathcal N}_{B,j}^{\bullet}
={}&\int q_j^*H_{B,j}^{\bullet*}s_{X,n}p_{X,n}^*\,dx
+\int H_{B,j}^{\bullet*}\dot q_jp_{X,n}^*\,dx\\
&+\int q_j^*\sum_{a_{-j}}
\pi_B(a_{-j}\mid n;\alpha)
\dot d_j^\bullet(a_{-j},x)p_{X,n}^*(x)\,dx.
\end{aligned}
\end{equation}
In particular, the policy channel (III) in
\eqref{eq:Dchannels}--\eqref{eq:Schannels} is absent.

The first two terms of \eqref{eq:NBchannels} are represented, respectively,
by
\[
q_j^*(X_i,n)H_{B,j}^{\bullet*}(X_i,n)
-\mathcal N_{B,j}^{\bullet}(0)
\quad\text{and}\quad
\{A_{ij}-q_j^*(X_i,n)\}H_{B,j}^{\bullet*}(X_i,n),
\]
by the covariate-score calculation and \eqref{eq:representer}; their sum is
$A_{ij}H_{B,j}^{\bullet*}(X_i,n)
-\mathcal N_{B,j}^{\bullet}(0)$.
Applying \eqref{eq:mrep} configuration by configuration to the last term
of \eqref{eq:NBchannels} gives an outcome-score representer
$R_{B,Y,j}^{\bullet}(O_i)$.  For DATT it is
$R_{Y,j}^{DATT}$ in \eqref{eq:RYD} with $\pi_\delta^*$ replaced by $\pi_B$,
and for SATT it is $R_{Y,j}^{SATT}$ in \eqref{eq:RYS} with the same
replacement.  Consequently,
\begin{equation}\label{eq:NBgrad}
\dot{\mathcal N}_{B,j}^{\bullet}
=\E_n\!\left[
\{A_{ij}H_{B,j}^{\bullet*}(X_i,n)
-\mathcal N_{B,j}^{\bullet}(0)
+R_{B,Y,j}^{\bullet}(O_i)\}S_n(O_i)
\right].
\end{equation}
The outcome representer has conditional mean zero given
$(A_i,X_i,N_i=n)$, so the expression inside braces in
\eqref{eq:NBgrad} has size-conditional mean zero.

The denominator calculation is unchanged:
\[
\dot{\mathcal D}_j
=\E_n[\{A_{ij}-p_j^*(n)\}S_n(O_i)].
\]
Using
$\mathcal N_{B,j}^{\bullet}(0)
=p_j^*(n)\tau_{B,j}^{\bullet*}(\alpha,n)$, the quotient rule and
\eqref{eq:NBgrad} yield
\[
\begin{aligned}
\left.\partial_t\tau_{B,j}^{\bullet}(\alpha,n;t)\right|_{t=0}
&=\frac{1}{p_j^*(n)}
\{\dot{\mathcal N}_{B,j}^{\bullet}
-\tau_{B,j}^{\bullet*}(\alpha,n)\dot{\mathcal D}_j\}\\
&=\E_n[\varphi_{B,j}^{\bullet*}(O_i;\alpha,n)S_n(O_i)],
\end{aligned}
\]
where the resulting expressions are exactly
\eqref{eq:phiBDATT} and \eqref{eq:phiBSATT}.

It remains to verify mean zero and tangent-space membership.  Multiplying
the unit-level candidate by $p_j^*(n)$ gives
\[
\begin{gathered}
p_j^*(n)\varphi_{B,j}^{\bullet*}=L_Y+L_A+L_X,\\
L_Y=R_{B,Y,j}^{\bullet},\qquad
L_A=\{A_{ij}-q_j^*(X_i,n)\}
\{H_{B,j}^{\bullet*}(X_i,n)
-\tau_{B,j}^{\bullet*}(\alpha,n)\},\\
L_X=q_j^*(X_i,n)
\{H_{B,j}^{\bullet*}(X_i,n)
-\tau_{B,j}^{\bullet*}(\alpha,n)\}.
\end{gathered}
\]
Here $L_Y$ has conditional mean zero given $(A_i,X_i,N_i=n)$,
$L_A$ is an $X_i$-measurable multiple of
$A_{ij}-q_j^*(X_i,n)$, and
\[
\E_n[L_X]
=\mathcal N_{B,j}^{\bullet}(0)
-p_j^*(n)\tau_{B,j}^{\bullet*}(\alpha,n)=0.
\]
Thus the candidate has size-conditional mean zero and belongs to
$\mathcal T_{j,n}$.  It represents the derivative along every regular
submodel and is therefore the size-$n$ unit-level canonical gradient.

Finally, averaging over units and differentiating the mixture over the
random cluster-size law exactly as in Step 3 gives
\[
\begin{aligned}
\E[\varphi_B^{\bullet*}(O_i;\alpha)S(O_i)]
={}&\sum_n\dot\nu_n\bar\tau_B^{\bullet*}(\alpha,n)
+\sum_n\nu_n\frac1n\sum_{j=1}^n
\dot\tau_{B,j}^{\bullet}(\alpha,n)\\
={}&\left.\partial_t\tau_B^\bullet(\alpha;t)\right|_{t=0}.
\end{aligned}
\]
Moreover, the size-conditional mean-zero property of every unit-level
gradient gives
\[
\begin{aligned}
\E[\varphi_B^{\bullet*}(O_i;\alpha)]
={}&\sum_n\nu_n\frac1n\sum_{j=1}^n
\E_n[\varphi_{B,j}^{\bullet*}(O_i;\alpha,n)]\\
&+\sum_n\nu_n
\{\bar\tau_B^{\bullet*}(\alpha,n)-\tau_B^{\bullet*}(\alpha)\}
=0.
\end{aligned}
\]
The within-size component is in the outcome, assignment, and covariate
tangent blocks, while
$\bar\tau_B^{\bullet*}(\alpha,N_i)-\tau_B^{\bullet*}(\alpha)$ belongs to
$\mathcal T_N$.  The expression in \eqref{eq:eifB} therefore has mean zero,
belongs to the full tangent space, and represents the global pathwise
derivative.  The same overlap, bounded-cluster-size, and moment conditions
used in Theorem~\ref{thm:eif} imply square integrability.  Hence
\eqref{eq:eifB} is the global efficient influence function, and its
variance is the corresponding semiparametric efficiency bound.
\hfill$\square$

\end{document}